\documentclass[12pt,letterpaper]{article}
\usepackage{fullpage}
\usepackage[top=2cm, bottom=3.5cm, left=2.0cm, right=2.0cm]{geometry}
\usepackage{amsmath,amsthm,amsfonts,amssymb,amscd}
\usepackage{lastpage}
\usepackage{enumerate}
\usepackage{fancyhdr}
\usepackage{mathrsfs}
\usepackage{xcolor}
\usepackage{graphicx,caption}
\usepackage{listings}
\usepackage{hyperref}
\makeatletter
\newcommand\xlabel[2][]{\phantomsection\def\@currentlabelname{#1}\label{#2}}
\makeatother

\usepackage{float}
\floatstyle{plaintop}
\restylefloat{table}

\usepackage{booktabs,caption}
\usepackage[flushleft]{threeparttable}

\usepackage{bm}
\usepackage[natbibapa]{apacite}

\usepackage{placeins}
\usepackage[utf8]{inputenc}
\usepackage[english]{babel}
\usepackage{natbib}
\usepackage{lscape}
\usepackage{subcaption}

\usepackage{booktabs}
\usepackage{multirow}
\usepackage{siunitx}

\usepackage{titling}

\usepackage{authblk}

\usepackage{afterpage}

\usepackage{thmtools, thm-restate}
\usepackage{setspace}

\usepackage{enumitem}

\usepackage{titling}

\hypersetup{%
   colorlinks=true,       
   citecolor=magenta,          
   linkcolor=magenta,          
}
 
\usepackage{mathtools}
\DeclarePairedDelimiter\floor{\lfloor}{\rfloor}

\newtheorem{assumption}{Assumption}

\newtheorem{theorem}{Theorem}
\newtheorem{cor}{Corollary}

\DeclareMathOperator*{\argmax}{arg\,max}
\DeclareMathOperator*{\argmin}{arg\,min}  
\DeclareMathOperator*{\plim}{plim}

\providecommand{\keywords}[1]
{
    \emph{Keywords--} #1
}

\begin{document}

\title{
Asymptotic properties of Bayesian inference \\
in linear regression with a structural break
}
\date{January 3, 2022}
\thanksmarkseries{alph}

\author{Kenichi Shimizu \thanks{E-mail address: 
Kenichi.Shimizu@glasgow.ac.uk. 
I am grateful for Andriy Norets, my dissertation advisor, for guidance and encouragement through the work on this project. I thank for valuable comments and suggestions from 
Eric Renault, 
Susanne Schennach, 
Dimitris Korobilis, 
Jesse Shapiro, 
Kenneth Chay, 
Toru Kitagawa, 
Adam McCloskey, 
Siddhartha Chib,  and
Florian Gunsilius.
I thank to audiences at the 2020 NBER-NSF Seminar on Bayesian Inference in Econometrics and Statistics (SBIES) conference for helpful discussions.
This work was supported by the Economics department dissertation fellowship at Brown University. All remaining errors are mine.

The paper was previously titled ``Structural break in linear regression models: Bayesian asymptotic analysis'' and ``Bayesian inference in linear regression with a structural break." 
 } }

\affil{\small{ \textit{ 
Adam Smith Business School \protect\\
University of Glasgow \protect\\ University Ave, Glasgow, G12 8QQ, United Kingdom
}}}
  


\begin{titlingpage}
\usethanksrule 

\maketitle

\begin{abstract}
This paper studies large sample properties of a Bayesian approach to inference about slope parameters $\gamma$ in linear regression models  with a structural break.
In contrast to the conventional approach to inference about $\gamma$ that does not take into account the uncertainty of the unknown break location $\tau$, 
the Bayesian approach that we consider incorporates such uncertainty.
Our main theoretical contribution is  a Bernstein-von Mises type theorem (Bayesian asymptotic normality) for $\gamma$ under a wide class of priors, which essentially indicates an asymptotic equivalence between  the conventional frequentist and  Bayesian inference. 
Consequently, a frequentist researcher could look at credible intervals of $\gamma$ to check robustness with respect to  the uncertainty of $\tau$.
Simulation studies show that the conventional confidence intervals of $\gamma$ tend to undercover in finite samples whereas the credible intervals offer more reasonable coverages in general. As the sample size increases, the two methods coincide, as predicted from our theoretical conclusion. 
Using data from Paye and Timmermann (2006) on stock return prediction, we illustrate that the traditional confidence intervals on $\gamma$ might underrepresent  the true sampling uncertainty. 
\end{abstract}


\end{titlingpage}

\onehalfspacing

\section{Introduction}
We consider the  linear regression  with a structural break, following the notations of \cite{bai1997}:
\begin{equation}\label{model}
    y_t=
    \begin{cases}
      w_t'\alpha+z_t'\delta_1+\epsilon_t, & \text{for}\ t=1,\ldots, \floor*{\tau T} \\
      w_t'\alpha+z_t'\delta_2+\epsilon_t, & \text{for}\ t=\floor*{\tau T}+1,\ldots, T,
    \end{cases}
\end{equation} 
where $w_t$ and $z_t$ are $d_w \times 1$ and $d_z \times 1$ vectors of covariates, and the random variable $\epsilon_t$ is a regression error. 
$\floor*{a}$ is the largest integer that is strictly smaller than $a$. 
The relationship between the outcome $y_t$ and the covariate $z_t$, measured by $\delta$'s, changes across regimes,
 which are defined by the break location parameter $\tau \in (0,1)$. 
There can be another set of covariates $w_t$ whose relationship with $y_t$, measured by $\alpha$, stays unchanged across the regimes.
%
%
%
The unknown parameters include the break location $\tau$ as well as the slope parameters $\gamma=(\alpha,\delta_1,\delta_2)$.
The focus of the current study is on inference about the slope parameter $\gamma$\footnote{
For an extensive review of 
important aspects in structural break models such as 
estimation and inference of
the number of breaks as well as break locations, see \cite{perron2006}.
}.

\subsection{The classic literature}
In the literature, the conventional least-squares estimators $\left( \hat{\tau}_{LS}, \hat{\gamma}_{LS} \right)$ for $\left( \tau, \gamma \right)$ are computed as follows: for each candidate $\tau$, compute the sum of squared residuals of the regression and denote the minimizing choice by $\hat{\tau}_{LS}$. Plug in the value $\tau = \hat{\tau}_{LS}$ in the model and define $\hat{\gamma}_{LS}=\hat{\gamma}(\hat{\tau}_{LS})$, where $\hat{\gamma}({\tau})$ is the usual OLS estimator of $\gamma$ assuming the break location $\tau$.
\cite{bai1997} assumes that the true jump size $\delta_0$ is either fixed or shrinks to zero as $T \to \infty$, but at a rate slower than $\sqrt{T} \to \infty$. 
Bai shows 
that $\hat{\tau}_{LS}$ converges at the rate $T^{-1}$ in the former case and, 
in the latter case,
finds an asymptotic distribution of $\hat{\tau}_{LS}$ that can be used for constructing confidence intervals for $\tau$.
In both cases, Bai proves that the asymptotic distribution of $\hat{\gamma}_{LS}$ is the same as that of $\hat{\gamma}(\tau_0)$, where $\tau_0$ is the true value of $\tau$. 
This means that one can ignore the very problem of unknown $\tau$ when making inference on $\gamma$.

Figure \ref{figure_posterior_tau} displays finite-sample distributions of $\hat{\tau}_{LS}$ (blue solid curves) which are produced based on 1,000 repeated experiments  on the following model $ y_t = \delta_0 1\left(t>\floor*{\tau_0T} \right)+\epsilon_t$\footnote{
Figure \ref{figure_posterior_tau} also shows distributions (red solid curve with small circles) of a Bayesian point estimator of $\tau$, the posterior mode. We later show that  the posterior mode converges to the same limiting distribution as $\hat{\tau}_{LS}$. 
}.
Note that despite the $T$-consistency, $\hat{\tau}_{LS}$ displays significant variation, especially when the true break size $\delta_0$ is small\footnote{In addition, the distributions exhibit three modes as reported in the literature (e.g.,\ \citealp{baek2021}; \citealp{casini_perron2021continuous}). }.
In practice, the conventional approach to inference on the slope parameters $\gamma$  would ignore this  uncertainty, neglecting all possible values of $\tau$ other than $\hat{\tau}_{LS}$. As a consequence, the  corresponding confidence intervals on $\gamma$ tend to undercover since it might not be the case that $\hat{\tau}_{LS} = \tau_0$ in a given sample (See our simulation in Section \ref{section-simulation}).

\FloatBarrier
\graphicspath{{Figures_final/tau_posterior/}}

\begin{figure}[hbt!]
\centering
\begin{subfigure}{.45\textwidth}
  \centering
  \includegraphics[width=1\linewidth]{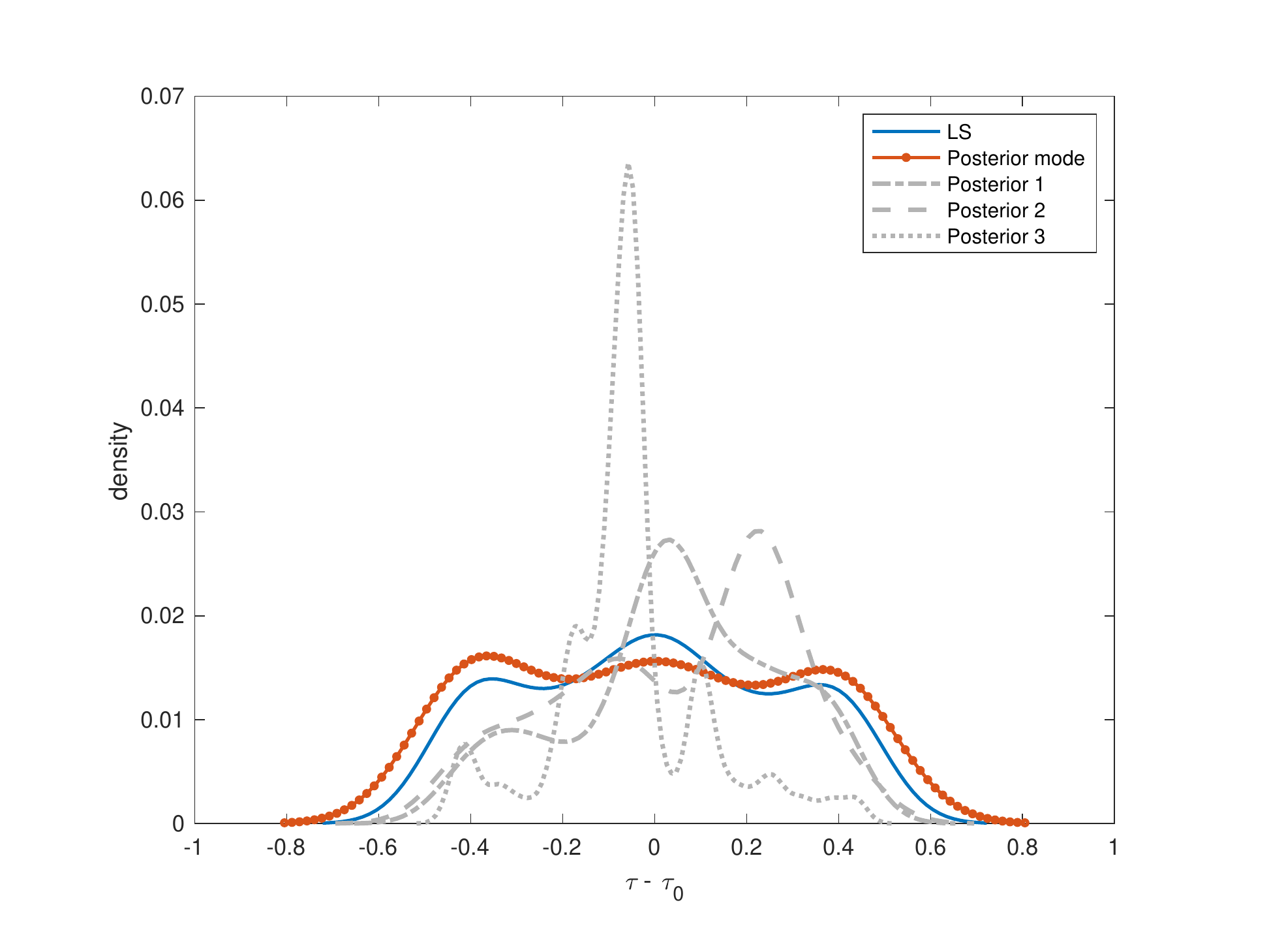}
  \caption{$\delta_0=0.25$}
\end{subfigure}
%
%
\begin{subfigure}{.45\textwidth}
  \centering
  \includegraphics[width=1\linewidth]{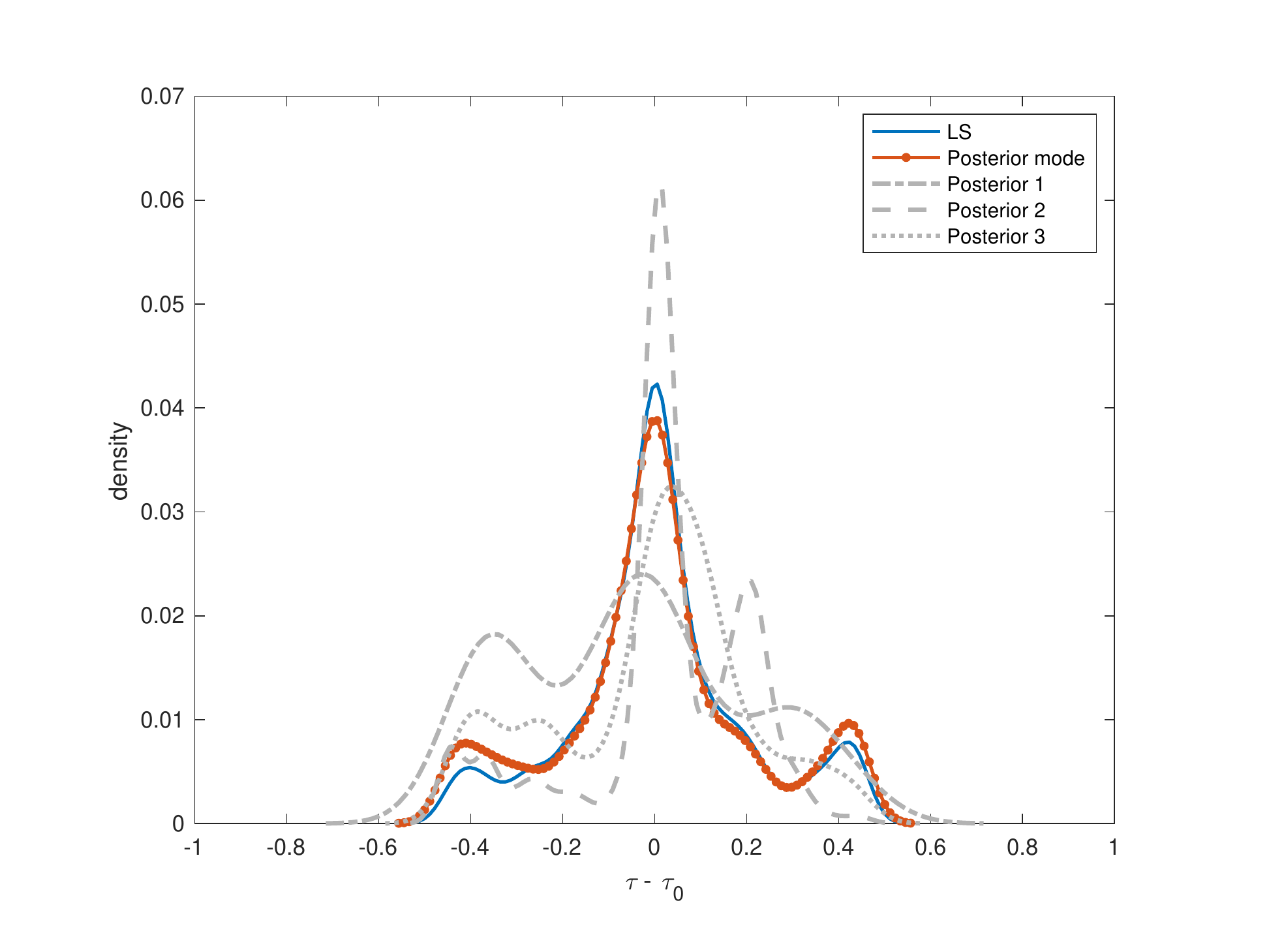}
  \caption{$\delta_0=0.5$}
\end{subfigure}
%
%
\caption{\footnotesize{
Finite-sample distribution of $\hat{\tau}_{LS}$ (blue solid curve) based on 
the model $ y_t = \delta_0 1\left(t>\floor*{\tau_0T} \right)+\epsilon_t$, $\epsilon_t \sim \text{i.i.d.} N(0,1), \tau_0=0.5, T=100$, with 1,000 repeated experiments.
The horizontal axis is $\tau-\tau_0$.
We also show finite-sample distribution of the posterior mode of $\tau$ (red solid curve with small circles).
In addition, we randomly chose 3 data realizations out of the 1,000 repetitions to plot posterior densities of $\tau$ in gray dashed curves (hence each of them represents a realization of one data set).}}
\label{figure_posterior_tau}
\end{figure}

\subsection{Bayesian perspective}
For a Bayesian, this non-standard estimation problem\footnote{
Estimation of structural break models is considered non-standard in a sense that there is a non-regular parameter (e.g.,\ break location) whose point estimator converges faster than $T^{-1/2}$, the rate at which the regular parameters (e.g.,\  slope coefficients) converge. 
}
can be dealt with  by placing prior on both $\tau$ and $\gamma$ and by computing  corresponding  posterior probabilities. The uncertainty of  $\tau$ will be automatically reflected on the marginal posterior probability of $\gamma$. This is because the posterior distribution of $\gamma$ given the data $\bm{D}_T$ can be written
as a mixture where the weights correspond to the marginal posterior density $\pi_T(\tau)$ for $\tau$:
\begin{equation}
p\left( \gamma |\bm{D}_T \right)=\int_0^1 p\left( \gamma |\tau,\bm{D}_T \right) \pi_T(\tau) d\tau, \label{mixture}
\end{equation}
where $p\left( \gamma |\tau,\bm{D}_T \right) $ is the posterior conditional distribution of $\gamma$ given $\tau$.
%
%
%
%
%
%
%
%
The posterior density $\pi_T(\tau)$ reflects the uncertainty of $\tau$ given the data set.
Figure \ref{figure_posterior_tau} shows  three realizations of $\pi_T(\tau)$ (gray dashed curves) which are randomly chosen out of the 1,000 repetitions. 
Compared to the conventional approach,  the key difference is that
 the Bayesian approach \eqref{mixture} incorporates all possibilities of $\tau$ (not just $\hat{\tau}_{LS}$) and weights them according to the posterior density. 
As we  see in simulation studies, this results in longer lengths of  Bayesian credible intervals of $\gamma$ compared to the conventional  counterparts. Consequently, the credible intervals tend to avoid  undercoverage. 
See Section \ref{section-simulation} for further discussion.
%
%
%
%
%
%
Note that, unlike  conventional frequentist methods, Bayesian inference has a valid interpretation even in finite samples as it does not rely on asymptotic theory.

In this study, we examine the asymptotic behavior of Bayesian estimation of the considered model under the fixed jump size framework. 
Specifically, we prove a Bernstein-von Mises type theorem for the slope parameters $\gamma$ which validates a frequentist interpretation of  Bayesian credible regions. 
%
A Bayesian researcher can invoke our theorem to convey  statistical results to  frequentist researchers. A frequentist researcher could look at the credible interval of $\gamma$ to check robustness with respect to the uncertainty of the break location. Such sensitivity analysis is reasonable as our result guarantees the credible interval to converge to the conventional confidence interval. 
%
%
We first establish theoretical results under normal likelihood and  natural conjugate prior. 
We further  extend the results to non-conjugate priors using  Laplace   approximations.
%



The literature on the theoretical properties of Bayesian approaches in non-regular models such as \eqref{model} is very scarce despite their popularity in applications.
To our knowledge, frequentist properties of the Bayesian approach for  linear regression models with structural breaks have not been studied in the literature. \cite{ghosal_samanta1995} consider a general non-regular estimation problem from a Bayesian perspective and establish conditions under which the Bernstein-von Mises theorem holds for the regular part of the parameter. However, their assumptions are difficult to verify in regard to our model in consideration.



Recently, \cite{casini_perron2020generalized} propose a generalized Laplace estimator of the break location $\tau$
which is defined by an integration rather than an optimization.
%
%
Their approach provides a better approximation about the  uncertainty in $\tau$ than the conventional method. 
Although our focus of the current paper is on  inference about the slope coefficients $\gamma$ and not  $\tau$, our Bayesian approach toward inference shares the same spirit; any statement about $\gamma$ is expressed as a weighted average \eqref{mixture} over the marginal posterior density of $\tau$.
%

The paper is organized as follows. 
Section \ref{section-model} introduces the model and lists a set of assumptions.
Section \ref{section-normal-default} introduces a Bayesian approach based on normal likelihood and conjugate prior.
The section then establishes frequentist properties  of the approach. 
Section \ref{section-general-theory} extends the results to non-conjugate priors.  
Section \ref{section-simulation} presents simulation evidence to assess  the adequacy of the asymptotic theory and to illustrate that  conventional confidence intervals on the slope parameters tend to undercover.
Section \ref{section-application} reports an empirical application to the stock return prediction model of  \citet{paye_timmermann2006}.
Section \ref{section-conclusion} concludes the paper.
The mathematical proofs and derivations are listed in the Appendix. 
Additional tables are provided in the online appendix. 

\section{The model and data generating process}\label{section-model}

\subsection{The model}
Using the reparametrization $x_t=(w_t',z_t')'$, $\beta=(\alpha',\delta_1')'$, and $\delta=\delta_2-\delta_1$, the equations \eqref{model} can be rewritten as 
\begin{equation}\label{model2}
    y_t=
    \begin{cases}
      x_t'\beta+\epsilon_t, & \text{for}\ i=1,\ldots, \floor*{\tau T} \\
      x_t'\beta+z_t'\delta+\epsilon_t, & \text{for}\ i=\floor*{\tau T}+1,\ldots, T.
    \end{cases}
\end{equation}
Note that $z_t$ is a subvector of $x_t$. More generally, let $z_t=R'x_t$, where $R$ is a $d_x \times d_z$ known matrix with full column rank and hence $z_t$ is defined as a linear transformation of $x_t$. For $R=(0_{d_z\times d_w},I_{d_z})'$, we obtain  model \eqref{model2}. For $R=I_{d_x}$, a pure change model is obtained. 
%
To rewrite the model in matrix form, we introduce further notations. Define 
$Y=(y_1,\ldots,y_T)'$, 
$\epsilon=(\epsilon_1,\ldots,\epsilon_T)'$, 
$X=(x_1,\ldots,x_T)'$, 
$X_{1\tau}=(x_1,\ldots,x_{\floor*{\tau T} },0,\ldots,0)'$, 
$X_{2\tau}=(0,\ldots,0,x_{\floor*{\tau T} +1},\ldots,x_T)'$.
Define $Z, Z_{1\tau}$, and $Z_{2\tau}$ similarly. Then, 
$Z=XR$, 
$Z_{1\tau}=X_{1\tau}R$,
and 
$Z_{2\tau}=X_{2\tau}R$.
Now, the equations \eqref{model2} can be written as
\begin{equation}\label{model3}
Y
= X\beta + Z_{2\tau}\delta + \epsilon 
= \chi_\tau \gamma +\epsilon,
\end{equation}
where $\chi_\tau = (X, Z_{2\tau})$ and $\gamma=(\beta',\delta')'$.
$S_T(\tau)$ denotes the sum of squared residuals of the regression \eqref{model3} given $\tau$. Let $\mathcal{H} \subset (0,1)$ be the space of the break locations. The least-squares estimator of $\tau$ is defined as	 
\begin{equation}
\hat{\tau}_{LS} =  \argmin_{\tau \in \mathcal{H}} S_T(\tau), \label{tau_hat}
\end{equation}
and the least-squares estimator for the slope coefficients $\gamma=(\beta',\delta')'$ is
\begin{equation}
 \hat{\gamma}_{LS} = \hat{\gamma}(\hat{\tau}_{LS}), \label{gamma_hat}
 \end{equation}
where $\hat{\gamma}(\tau)$ denotes the usual OLS estimator given the value of $\tau$.

%
%

\subsection{Data generating process}\label{section-dgp}

The data are assumed to include $T$ observations on a response and a vector of covariates: $\bm{D}_T=( \bm{Y}_T, \bm{X}_T) = (y_1,\ldots,y_T,x_1,\ldots,x_T)$
where $y_t \in  \mathbb{R}$ and $x_t \in \mathcal{X} \subset \mathbb{R}^{d_x}$, $t=1,...,T$. 
$\mathcal{X}$ is assumed to be a convex and bounded set.
Conditional on $ \bm{X}_T $, the response is generated according to  model \eqref{model3} with the true parameters $(\gamma_0',\sigma_0^2, \tau_0)'$. We use $\theta=(\gamma',\sigma^2)'$ to denote the regression parameters. We make the following assumptions about the true data-generating-process (DGP): 
\pagebreak

\begin{assumption}\label{assumption1}
\quad 
\begin{enumerate}[label=(\roman*)]
\item 
$\delta_0 \ne 0$. 
\item
$\epsilon_t$ is i.i.d. with $E(\epsilon_t|x_t)=0$, $E(\epsilon_t^2|x_t) =\sigma_0^2$, where $\sigma_0^2$ is unknown to the econometrician.
\item $\Sigma_X =E[x_tx_t']= \plim \frac{1}{T} \sum_{t=1}^T x_t x_t'$ exists and is positive definite. 
\item For all $\tau_1,\tau_2 \in (0,1)$ with $\tau_1 < \tau_2$, 
$\frac{1}{T} \sum_{\floor*{\tau_1T}+1}^{\floor*{\tau_2T}} x_t\epsilon_t =O_p(T^{-1/2})$ 
and 
$\frac{1}{T} \sum_{\floor*{\tau_1T}+1}^{\floor*{\tau_2T}} x_tx_t' =(\tau_2-\tau_1)\Sigma_X+O_p(T^{-1/2})$ 
\end{enumerate}

\end{assumption}

Under the above assumptions, the classical theoretical results apply. \cite{bai1997} shows that the convergence rate of $\hat{\tau}_{LS}$ is $T^{-1}$ if $\delta_0$ is fixed with respect to the sample size:
\[
\hat{\tau}_{LS} =\tau_0 + O_p(T^{-1}),
\]
and that the least-squares estimator for $\gamma$ is asymptotically normal with the asymptotic covariance matrix being the same as if $\tau_0$ is known: 
\begin{equation}
\sqrt{T}\left(  \hat{\gamma}_{LS} -\gamma_0 \right)
\overset{d}{\to}
N_{(d_x+d_z)}(0, \sigma_0^2 V^{-1}),
\label{conventional_theory}
\end{equation}
where 
\[
V
=\plim T^{-1}
\begin{pmatrix}
\sum_{t=1}^T x_tx_t' & \sum_{t=\floor*{\tau_0T}+1}^T x_tz_t' \\
 \sum_{t=\floor*{\tau_0T}+1}^T z_tx_t'  &  \sum_{t=\floor*{\tau_0T}+1}^T z_tz_t' 
\end{pmatrix}
=
\plim T^{-1}
\chi_{\tau_0}'\chi_{\tau_0}.
\]
This means that $\tau$ can be treated as known for the purpose of inference about $\gamma$. In other words, the  uncertainty of the break location  is essentially ignored, and thus the confidence interval for $\gamma$ tends to undercover (see Section \ref{section-simulation} for simulation).

There are several comments on Assumption \ref{assumption1}. 
In threshold regression models (see \citealp{hansen2000}), the threshold variable is often one of the regressors. 
In this case, sorting the threshold variable leads to a trend in the regressors, which requires an alternative approach for the asymptotic analysis. 
We do not consider the case with one of the regressors being the threshold variable in this paper. 
In addition, we require the regression errors to be i.i.d.\ with variance $\sigma^2$. 
Adding more flexibility such as heteroscedasticity and serial correlation would be an important future direction.



\section{A Bayesian approach under normal likelihood and conjugate prior}\label{section-normal-default} 

The distribution of covariates is assumed to be ancillary and it is not modeled.
Throughout this paper, we assume the normal likelihood function\footnote{
Similarly, \cite{qu_perron2007} propose  a quasi-maximum likelihood estimator assuming normal errors.
}
\begin{equation}
p(\bm{Y}_T |\bm{X}_T, \theta,\tau) 
=\prod_{t=1}^T 
 \frac{1}{\sqrt{2 \pi \sigma^2}} 
 \exp \left( -\frac{\left( y_t - \chi_{\tau,t}' \gamma \right)^2  }{2\sigma^2}  \right), \label{normal_like}
\end{equation}
where $\chi_{\tau,t}'$ is the $t$th row of the matrix $\chi_{\tau}$.
Note that the normality is not assumed for the true DGP, so the model can be mis-specified.


%
The break location $\tau$ and the regression parameters $\theta$ are independent a-priori and the prior on $\theta$ is  the natural conjugate prior. 
That is, $ \pi \left(\gamma,\sigma^2, \tau \right) = \pi(\gamma | \sigma^2) \pi(\sigma^2) \pi(\tau)$ where 
the prior on $\gamma$ conditional on $\sigma^2$ is normal $N_{(d_x+d_z)}(\underline{\mu}, \sigma^2 \underline{H}^{-1}) $
and 
the prior on $\sigma^2$ is inverse-gamma $InvGamma(\underline{a},\underline{b}) $.
Note that by taking $\underline{H} \to 0$, $\underline{a} \to -(d_x+d_z)/2$, and $\underline{b} \to 0$, we have  the uninformative improper prior $ \pi \left(\gamma,\sigma^2 \right) \propto \sigma^{-2}$ as a special case.
The prior on $\tau$ can be of any form as long as it is positive at $\tau_0$, and $\pi(\tau)$ is finite for all $\tau \in \mathcal{H} $.

The  conjugate prior is a popular choice in the Bayesian estimation of linear regression models. 
Our restriction on the prior for the break location $\tau$ is very mild. 
For example, the uniform distribution on $\mathcal{H}$ satisfies the requirement. 
Recently, \cite{baek2021} investigates the same model \eqref{model}. As the distribution of $\hat{\tau}_{LS}$ might exhibit tri-modality for small jumps, Baek proposes a new point estimator for $\tau$ based on a modified objective function. The proposed modification can be regarded as equivalent to specifying a certain type of prior for $\tau$ and indeed such prior satisfies our restriction.


Under the normal likelihood function and the prior defined above, the posterior distributions are
\begin{align}
\pi(\tau | \bm{D}_T ) &\propto 
\left[ 
\det
\left( \bar{H}_\tau \right)
 \right]^{-0.5} 
\bar{b}_\tau^{-\bar{a}} \times \pi(\tau), \label{posterior_tau_conjugate} \\
\gamma
\big| \tau, \bm{D}_T &\sim t_{(d_x+d_z)} 
\left(
2\bar{a}, 
\bar{\mu}_{\tau},
( \bar{b}_\tau /  \bar{a} )\bar{H}_\tau^{-1}
\right), \label{posterior_gamma_conjugate} \\
\sigma^2| \tau, \bm{D}_T &\sim Inv Gamma \left( \bar{a},\bar{b}_\tau \right), \label{posterior_sigma_conjugate}
\end{align}
where 
$ \bar{H}_\tau =  \underline{H} + \chi_\tau'\chi_\tau$, 
$\bar{\mu}_{\tau} = \bar{H}_\tau^{-1} \left[  \underline{H}\underline{\mu}+\chi_\tau' Y   \right]$,
$\bar{b}_\tau = \underline{b} +0.5\left[ \underline{\mu}' \underline{H} \underline{\mu} + Y'Y - \bar{\mu}_\tau' \bar{H}_\tau \bar{\mu}_\tau \right]$, and 
$\bar{a} = \underline{a}+ T/2$, and 
$t_k(v,\mu,\Sigma)$ is the $k$-dimensional t-distribution with $v$ degrees of freedom, a location vector $\mu \in \mathbb{R}^k$, and a $k \times k$ shape matrix $\Sigma$.
See Appendix \ref{derivation} for the derivation.


Due to the availability of the closed-forms for the conditional posteriors given $\tau$, the posterior sampling is simple and fast. One can first draw $\tau_{(1)},\ldots,\tau_{(S)}$ from the marginal posterior of $\tau$ as in  \eqref{posterior_tau_conjugate} via, for example, the Metropolis-Hastings algorithm, where $S$ is the number of posterior draws. For each $\tau_{(s)}$, one can sample posterior draws of $\sigma^2_{(s)}$ from the posterior conditional on $\tau=\tau_{(s)}$, namely \eqref{posterior_sigma_conjugate}. Conditional on $\tau$ and $\sigma^2$, one can draw $\gamma$ from $p(\gamma | \sigma^2, \tau, \bm{D}_T )$\footnote{
It can be shown that 
$\gamma
\big| \sigma^2, \tau, \bm{D}_T \sim N_{(d_x+d_z)} 
\left(
\bar{\mu}_\tau,
\sigma^2
\bar{H}_\tau^{-1} 
\right)
$}. 
For example, a laptop with a 2.2GHz processor and 8GB RAM takes about 4.1 seconds to draw 10,000 posterior draws in an empirical example in Section \ref{section-application} that has ten slope coefficients in total. 

\subsection{Asymptotic theory}\label{section-normal-default-theory}

We investigate the asymptotic behavior of the Bayesian method under the normal likelihood and the conjugate prior defined above. We do so in two steps. Section \ref{SectionConsistency} shows that the marginal posterior of the break location $\tau$ contracts to the true value $\tau_0$ at the rate of $T^{-1}$, the same rate at which the least-squares estimator $\hat{\tau}_{LS}$ converges. The proof is based on studying the behavior of the log ratio of the marginal posterior densities of $\tau$. In addition, we establish the limiting distribution of the posterior mode of $\tau$. Section \ref{SectionBvm} establishes a Bernstein-von Mises type theorem for the regression slope coefficients $\gamma$. 
The proof is based on the $T$-consistency of the marginal posterior of $\tau$ and the fact that the conditional posterior for $\sqrt{T}\left(  \gamma - \hat{\gamma}_{LS} \right)$ given $\tau$ is asymptotically normal. 
Proofs of the theorems can be found in Appendix \ref{AppendixThms}.

\subsubsection{Marginal posterior of $\tau$}\label{SectionConsistency}

An intermediate step for proving the Bernstein-von Mises theorem is the marginal posterior consistency of $\tau$ at rate $T^{-1}$.  Marginal posteriors have not been studied extensively or systematically in the literature. Here, we directly analyze the form of the marginal posterior of $\tau$. 
Let $L_T(\tau)$ be the marginal likelihood conditional on $\tau$, that is 
\[
L_T(\tau)=
\int p(\bm{Y}_T |\bm{X}_T, \theta,\tau) \pi(\theta, \tau) d\theta,
\]
%
which is available up to a multiplicative constant under the normal likelihood and the conjugate prior as can be seen in \eqref{posterior_tau_conjugate}.
The marginal posterior density $\pi_T(\tau) $ of $\tau$ is defined as 
\[
\pi_T(\tau) 
= \frac{ L_T(\tau)  }{ \int  L_T(\tau) d\tau }.
\]

The following theorem establishes the first step for proving the Bernstein-von Mises theorem, the $T$-consistency of the marginal posterior of $\tau$.
It states that the posterior mass outside of a ball around $\tau_0$ with radius proportional to $T^{-1}$ will be asymptotically negligible. 
\begin{theorem}[Marginal posterior consistency of $\tau$ at rate $T^{-1}$]\label{ThmConsistency}
Suppose  Assumption \ref{assumption1} holds. 
Then, under the normal likelihood 
and the conjugate prior described above,  
$\forall \eta>0, \epsilon>0$, $\exists M>0$ and $k>0$ such that $T \geq k \implies$
\[
P_{\theta_0, \tau_0} \left( 
\int_{B^c_{M/T} (\tau_0)} \pi_T(\tau) d\tau < \eta
\right)
>
1-\epsilon,
\]
where for any constant $d>0$, $B^c_{d}(\tau_0)$ denotes the set difference $\mathcal{H} \setminus (\tau_0-d, \tau_0+d)$. 

\end{theorem}

%
The proof of Theorem 1 is built on some intermediate steps, Propositions \ref{PropRatio}-\ref{PropACts}.
It can be shown that 
$
\int_{ B^c_{M/T}(\tau_0) }  \pi_T(\tau) d\tau 
$ 
is bounded by the product of 
$
\int_{ B^c_{M/T}(\tau_0) } \frac{ L_T( \tau)   }{ L_T(\tau_0)  } d\tau  
$
and the inverse of
$ 
\int_{  B^c_{M_0/T} (\tau_0) } \frac{ L_T(\tau') }{ L_T( \tau_0 )} d\tau'   
$
for each $T$ and for any $M_0>0$.
Proposition \ref{PropRatio} shows that under the normal likelihood and the conjugate prior, 
due to the availability of the marginal likelihood conditional on $\tau$ up to a normalization constant as in \eqref{posterior_tau_conjugate},
studying 
the log marginal likelihood ratio
boils down to comparing the sum of squared residuals  $S_T(\tau)$.
Proposition \ref{PropLimitQ} establishes the probability limit of $T^{-1} S_T(\tau)$, for which we show examples in Figure \ref{fig_Qn}.
We then show that the limit of $T^{-1} S_T(\tau)$ achieves a unique minimum at $\tau_0$ (Proposition  \ref{PropAMax}), 
and study 
the modulus of continuity of an appropriate empirical process (Proposition \ref{PropACts})
in order to derive  bounds.
The detail of the proof of Theorem 1 can be found 
in Appendix \ref{ProofThmConsistency}.
\FloatBarrier
\graphicspath{{Figures_final/}}
\begin{figure}[!htb]
\center
  \includegraphics[width=0.5 \linewidth]{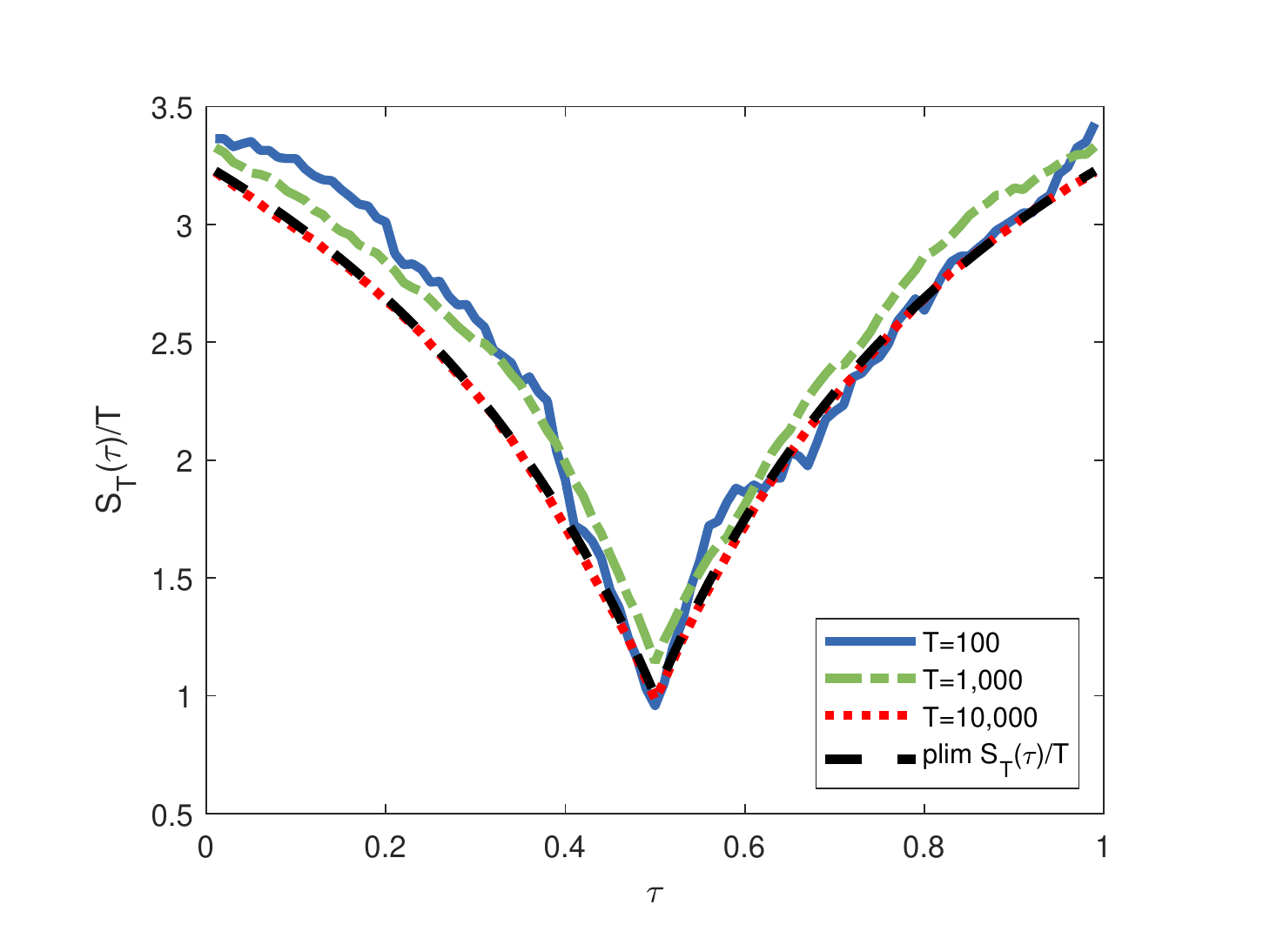}
\caption{\small{Example of $T^{-1} S_T(\tau)$ with
 $T=100$ (solid, blue), $T=1,000$ (dash-dotted, green), and $T=10,000$ (dotted, red) 
 and $\plim T^{-1} S_T(\tau)$ (dashed, black)
 }}\label{fig_Qn}
\end{figure}
\FloatBarrier

The Bayesian counterpart of the least-squares estimator $\hat{\tau}_{LS}$ would be the posterior mode:
\[
\hat{\tau}_{Bayes} = \argmax_{\tau \in \mathcal{H}} \pi_T(\tau).
\]
\cite{bai1997} shows that 
$ \argmax_{m} W^*(m) $
is the asymptotic distribution of $\hat{\tau}_{LS}$  
\footnote{
$W^*(m)$ is a stochastic process defined on the set of integers as follows: $W^*(0)=0$, $W^*(m)=W_1(m)$ for $m<0$, and $W^*(m)=W_2(m)$ for $m>0$, with 
\begin{alignat}{4}
W_1(m)& = -\delta_0 \sum_{i=m+1}^0 z_iz_i' \delta_0 &&+ 2\delta_0 \sum_{i=m+1}^0 z_i\epsilon_i, &&\text{ for } m=-1,-2,... \nonumber \\
W_2(m)& = -\delta_0 \sum_{i=1}^m z_iz_i' \delta_0 &&- 2\delta_0 \sum_{i=1}^m z_i\epsilon_i, &&\text{ for } m=1,2,... \nonumber
\end{alignat}
}.
A consequence of the proof of Theorem \ref{ThmConsistency} is that $\hat{\tau}_{Bayes} $ converges to the same limiting distribution.
%
%
See Appendix \ref{ProofThmMode} for a proof. 
\begin{cor}[Limiting distribution of the posterior mode of $\tau$]\label{ThmMode}
Suppose  Assumption \ref{assumption1} holds.
Then, under the normal likelihood 
and the conjugate prior described above,
\[
\floor*{T\left( \hat{\tau}_{Bayes} - \tau_0 \right) }   \overset{d}{\to} \argmax_{m} W^*(m) .
\]
\end{cor}

\subsubsection{Bernstein-von Mises Theorem for $\gamma$}\label{SectionBvm}
The marginal posterior of $\gamma$ is a mixture with weights corresponding to the marginal posterior density $\pi_T(\tau) $. Furthermore, due to Theorem \ref{ThmConsistency}, we can focus our attention on the values of $\tau$ in a $T^{-1}$ neighborhood of $\tau_0$:
\[
\int p( \gamma |\tau,\bm{D}_T ) \pi_T(\tau) d\tau 
=
\int_{B_{M/T}(\tau_0)} p( \gamma|\tau,\bm{D}_T) \pi_T(\tau) d\tau +o_p(1).
\]
We are now ready to establish the  Bernstein-von Mises type result. 
\begin{theorem}[Bernstein-von Mises theorem for the slope coefficients]\label{ThmBvm}
Suppose  Assumption \ref{assumption1} holds. 
Then, under the normal likelihood 
and the conjugate prior described above,
\[
d_{TV} \left(  
\pi \left[  
\begin{matrix}
\sqrt{T}\left( \gamma-  \hat{\gamma}_{LS}  \right)
\end{matrix}
\bigg| \bm{D}_T\right], 
N_{(d_x+d_z)}\left(  0,  \sigma^2_0V^{-1} \right) 
\right) \to 0 ,
\]
in $P_{\theta_0, \tau_0}-probability$ where $d_{TV}$ is the total variation distance. 

\end{theorem}
The proof of Theorem 2 exploits the fact that the conditional posterior for $\sqrt{T}\left(  \gamma - \hat{\gamma}_{LS} \right)$ given $\tau$ is asymptotically normal, which is close to the asymptotic distribution of $ \hat{\gamma}_{LS}$ when $\tau$ is close to $\tau_0$. A bound on the Kullback–Leibler (KL) divergence between two normal densities together with the $T$-consistency is used to make the argument precise. 
The proof is presented in Appendix \ref{ProofThmBvm}.

\section{An extension to non-conjugate priors}\label{section-general-theory}
The previous section establishes the asymptotic properties of the posterior distributions under the conjugate prior. 
A natural question is whether these results can be extended to other priors. 
For example, an independent prior between the slope coefficients $\gamma$ and the error variance $\sigma^2$, e.g.,\
$\pi(\gamma,\sigma^2)=\pi(\gamma)\pi(\sigma^2)$ with $\gamma \sim N_{(d_x+d_z)}(\underline{\mu}, \underline{\Sigma})$ and $\sigma^2 \sim InvGamma(\underline{a},\underline{b})$,
 is a popular choice for the Bayesian estimation of regression models in practice. 
Under the normal likelihood and the conjugate prior, the analytical expressions of the marginal posterior of $\tau$ up to a normalization constant \eqref{posterior_tau_conjugate}
and the conditional posterior of $\gamma$ given $\tau$ \eqref{posterior_gamma_conjugate}
facilitate the theoretical analysis. They are not available, for instance, under the independent prior mentioned above. 
In this section, we extend the theoretical results by keeping the normal likelihood \eqref{normal_like} but without requiring  the  conjugate prior on $\theta$.
In order to study the asymptotic behavior of the posterior distributions without having their closed-form expressions,
we employ Laplace approximation type results in \cite{hong_preston2012}.
To do so, we make an additional assumption as shown below. 
Let $\hat{\theta}(\tau)$ be the maximum likelihood estimator of $\theta$ conditional on $\tau \in \mathcal{H}$, i.e.,\ 
$\hat{\theta}(\tau)=\arg\sup_{\theta \in  \Theta} \log p(\bm{Y}_T | \bm{X}_T, \theta, \tau)$.
Denote by $\theta^*(\tau)$  the corresponding pseudo true parameter value
that minimizes the KL divergence between the model $p(\bm{Y}_T | \bm{X}_T, \theta, \tau)$
and the DGP. 

\begin{assumption}\label{assumption2}
\quad 
\begin{enumerate}[label=(\roman*)]
\item
There is a compact convex subset $\Theta$ of $\mathbb{R}^{d_x+d_z+1}$ 
such that $\theta^*(\tau) \in int\left(\Theta \right)$  for all $\tau \in \mathcal{H}$.
\item The prior $\pi(\theta,\tau)$ is supported on $\Theta \times \mathcal{H}$. It is continuous in $\theta$ and bounded away from 0 and $\infty$ around $\left( \theta^*(\tau), \tau \right)$  for all $\tau \in \mathcal{H}$.
\end{enumerate}
\end{assumption}
\noindent Under the normal likelihood and Assumption \ref{assumption1}, together with Assumption \ref{assumption2}, 
we can invoke the Laplace approximation results of \cite{hong_preston2012}.
Note that, under  the normal likelihood and Assumption \ref{assumption1},  $\theta^*(\tau)$ exists and is a function of parameters in the DGP.
%
In this section, we no longer assume the natural conjugate prior on $\theta$.
For instance, the independent prior $\pi(\gamma,\sigma^2,\tau)=\pi(\gamma)\pi(\sigma^2) \pi(\tau)$  mentioned above
satisfies the conditions in (ii) of  Assumption \ref{assumption2}
 as long as they are truncated on $\Theta$ and $\pi(\tau)$ is positive and finite at all $\tau$.
%

Theorem 3 below establishes the $T$-consistency of the marginal posterior of $\tau$ under this prior and the additional assumption. 
\begin{theorem}[Marginal posterior consistency of $\tau$ at rate $T^{-1}$, non-conjugate priors]\label{ThmConsistencyGeneral}
Suppose  Assumptions \ref{assumption1} and \ref{assumption2} hold. Then, under the normal likelihood,  
$\forall \eta>0, \epsilon>0$, $\exists M>0$ and $k>0$ such that $T \geq k \implies$
\[
P_{\theta_0, \tau_0} \left( 
\int_{B^c_{M/T} (\tau_0)} \pi_T(\tau) d\tau < \eta
\right)
>
1-\epsilon,
\]
where for any constant $d>0$, $B^c_{d}(\tau_0)$ denotes the set difference $\mathcal{H} \setminus (\tau_0-d, \tau_0+d)$. 
\end{theorem}
\noindent Recall that while proving the $T$-consistency  under  the conjugate prior (i.e.,\ Theorem 1),  
we utilize
the closed-form expression of the marginal posterior of $\tau$ up to a multiplicative constant \eqref{posterior_tau_conjugate} 
in order to study the behavior of the marginal likelihood ratio conditional on $\tau$. 
Under   non-conjugate priors, such expression is not available in general.
For this reason, we invoke a Laplace approximation   to investigate the quantity 
$\int p(\bm{Y}_T | \bm{X}_T, \theta, \tau) \pi(\theta,\tau) d\theta$ to prove Theorem 3. See Appendix \ref{ProofThmConsistencyGeneral} for the detail.

As in the previous section, an implication of the $T$-consistency of the marginal posterior of $\tau$ is 
that the posterior mode converges to the limiting distribution of $\hat{\tau}_{LS}$. Proof is in Appendix \ref{ProofThmModeGeneral}.
\begin{cor}[Limiting distribution of the posterior mode of $\tau$, non-conjugate priors]\label{ThmModeGeneral}
Suppose  Assumptions \ref{assumption1} and \ref{assumption2} hold. Then, under the normal likelihood,  
\[
\floor*{T\left( \hat{\tau}_{Bayes} - \tau_0 \right)}  \overset{d}{\to} \argmax_{m} W^*(m),
\]
where the stochastic process $W^*(m)$ is defined in Section \ref{SectionConsistency}. 
\end{cor}

Theorem 4 establishes our main theoretical result, the Bernstein-von Mises theorem for $\gamma$, under the prior defined in Assumption \ref{assumption2} (ii).
\begin{theorem}[Bernstein-von Mises theorem for the slope coefficients, non-conjugate priors]\label{ThmBvmGeneral}
Suppose  Assumptions \ref{assumption1} and \ref{assumption2} hold. Then, under the normal likelihood,  
\[
d_{TV} \left(  
\pi \left[  
\begin{matrix}
\sqrt{T}\left( \gamma-  \hat{\gamma}_{LS}  \right)
\end{matrix}
\bigg| \bm{D}_T\right], 
N_{(d_x+d_z)}\left(  0,  \sigma^2_0V^{-1} \right) 
\right) \to 0 ,
\]
in $P_{\theta_0, \tau_0}-probability$ where $d_{TV}$ is the total variation distance. 
\end{theorem}
When proving the corresponding result under the conjugate prior (i.e.,\ Theorem \ref{ThmBvm}), 
we utilize the closed-form expression of the marginal posterior of $\gamma$ given $\tau$ \eqref{posterior_gamma_conjugate}.
As this is not available under the prior in this section, we again use a Laplace  approximation to study the asymptotic behavior of the marginal posterior. 
See Appendix \ref{ProofThmBvmGeneral} for a proof.

\pagebreak 
\section{Simulation}\label{section-simulation}
The main purpose of the simulation studies below is to compare inference on the slope parameters $\gamma$ between the two methods: the conventional least-squares method in \cite{bai1997} and the Bayesian approach described in our paper.
For the Bayesian approach, we use the uniform prior for $\tau$ and the conjugate prior for the regression parameters with $\underline{H}=0.1 I_{(d_x+d_z)}$, $\underline{\mu}=0_{(d_x+d_z)}$, and $\underline{a}=\underline{b}=1$. The findings are similar even when we use the uninformative improper prior. 
Following the literature (e.g.,\ \citealp{casini_perron2021continuous}), we set the range of the candidate values of $\tau$ to be  $(\epsilon,1-\epsilon)$ with $\epsilon=0.05$ for all methods\footnote{It prevents the break location estimator from being in the first and last $100\epsilon$\% of the sample. The trimming parameter $\epsilon$ should not be chosen too high otherwise it might introduce bias in the break location estimate. \cite{casini_perron2021continuous} find the choice $\epsilon=0.05$ performs well in general, which we also confirm in our simulation exercises.}.

We consider the following model: $y_t=\delta_01(t>\floor*{\tau_0T})+\epsilon_t$.
In order to compare the methods in repeated experiments, for each combination of $\tau_0$, $\delta_0$, and $T$, we generate 1,000 data sets. 
We consider different values of the break location $\tau_0 \in \{0.3,0.5\}$, the jump size $\delta_0 \in \{0.25,0.5,1.0,2.0\}$, and the sample size $T \in \{20, 50,100,250, 500, 1000\}$. 
The error $\epsilon_t$ is independently and identically generated from $N(0,1)$.
In the online appendix, we present a robustness check with the errors generated from a mixture of two normals $0.5N\left(-1/\sqrt{2},1/2 \right) + 0.5N\left( 1/\sqrt{2},1/2 \right) $ and illustrate that the overall findings are similar to these under the normal DGP. 

Table \ref{table_simulation_delta} shows the simulation results concerning $\delta$.
The top panel ``Coverage" shows empirical coverages of the true jump size $\delta_0$ by the 95\% confidence and credible intervals. 
The frequentist confidence intervals are computed based on the conventional asymptotic theory \eqref{conventional_theory}. 
For the Bayesian approach, we report the equal-tailed credible intervals. 
The middle panel ``Length" presents the average lengths of the aforementioned intervals. 
The bottom panel ``MSE for $\delta$'' shows the mean-squared-errors for the point estimator, which is the least-squares estimator $\hat{\delta}_{LS}$ defined in \eqref{gamma_hat} for the conventional method and the posterior mean for the Bayesian approach. 

There are several significant findings.
First, for small $T$ and/or small $\delta_0$, the conventional  confidence intervals significantly undercover. Meanwhile, the Bayesian credible intervals have relatively reasonable coverages. 
Second, the Bayesian intervals tend to be longer than the conventional confidence intervals for small $T$ and/or $\delta_0$.
Third,  as $T$ increases, the discrepancy between the two methods decreases, as expected from the Bernstein-von Mises theorem that we establish.

\FloatBarrier
\begin{table}[!htb]

\begin{subtable}{0.5\textwidth}
\centering
\scalebox{0.55}{        
\begin{tabular}{lllllllllllll}
    \toprule
    \toprule
    \multirow{1}{*}{} &
      \multicolumn{4}{c}{Least-squares} &
      \multicolumn{4}{c}{Bayesian} \\
      \cmidrule(lr){2-5}\cmidrule(lr){6-9}
      $\delta_0=$ & {0.25} & {0.50} & {1.00} &{2.00} & {0.25} & {0.50} & {1.00} & {2.00} \\
      \midrule
      Coverage\\
$T=20$ & 0.68 & 0.77 & 0.88 & 0.95 & 0.96 & 0.97 & 0.96 & 0.95 \\
$T=50$ & 0.67 & 0.84 & 0.94 & 0.96 & 0.96 & 0.97 & 0.96 & 0.95 \\
$T=100$ & 0.69 & 0.90 & 0.96 & 0.95 & 0.96 & 0.97 & 0.96 & 0.95 \\
$T=250$ & 0.83 & 0.94 & 0.94 & 0.96 & 0.96 & 0.96 & 0.94 & 0.96 \\
$T=500$ & 0.91 & 0.95 & 0.95 & 0.96 & 0.97 & 0.96 & 0.96 & 0.96 \\
$T=1000$ & 0.93 & 0.94 & 0.95 & 0.96 & 0.96 & 0.95 & 0.95 & 0.96 \\
\hline 
Length\\
$T=20$ & 3.87 & 3.60 & 3.20 & 2.82 & 4.85 & 4.59 & 4.20 & 3.16 \\
$T=50$ & 2.31 & 2.07 & 1.82 & 1.76 & 2.91 & 2.67 & 2.13 & 1.79 \\
$T=100$ & 1.61 & 1.38 & 1.26 & 1.24 & 2.10 & 1.78 & 1.34 & 1.25 \\
$T=250$ & 0.93 & 0.81 & 0.78 & 0.78 & 1.21 & 0.92 & 0.80 & 0.79 \\
$T=500$ & 0.61 & 0.56 & 0.55 & 0.55 & 0.76 & 0.58 & 0.56 & 0.56 \\
$T=1000$ & 0.41 & 0.39 & 0.39 & 0.39 & 0.46 & 0.40 & 0.39 & 0.40 \\
\hline
MSE for $\delta$\\
$T=20$ & 3.85 & 2.79 & 1.75 & 0.60 & 1.13 & 0.91 & 0.86 & 0.58 \\
$T=50$ & 1.35 & 0.78 & 0.26 & 0.20 & 0.42 & 0.33 & 0.25 & 0.20 \\
$T=100$ & 0.67 & 0.28 & 0.11 & 0.10 & 0.21 & 0.15 & 0.11 & 0.10 \\
$T=250$ & 0.18 & 0.05 & 0.04 & 0.04 & 0.08 & 0.05 & 0.04 & 0.04 \\
$T=500$ & 0.05 & 0.02 & 0.02 & 0.02 & 0.03 & 0.02 & 0.02 & 0.02 \\
$T=1000$ & 0.02 & 0.01 & 0.01 & 0.01 & 0.01 & 0.01 & 0.01 & 0.01 \\
\bottomrule
\end{tabular}          
}
\caption{$\tau_0=0.5$}
\end{subtable}
    \hfill
\begin{subtable}{0.5\textwidth}
\centering
\scalebox{0.55}{        
\begin{tabular}{lllllllllllll}
    \toprule
    \toprule
    \multirow{1}{*}{} &
      \multicolumn{4}{c}{Least-squares} &
      \multicolumn{4}{c}{Bayesian} \\
      \cmidrule(lr){2-5}\cmidrule(lr){6-9}
      $\delta_0=$ & {0.25} & {0.50} & {1.00} &{2.00} & {0.25} & {0.50} & {1.00} & {2.00} \\
      \midrule
      Coverage\\
$T=20$ & 0.66 & 0.70 & 0.87 & 0.93 & 0.96 & 0.97 & 0.96 & 0.95 \\
$T=50$ & 0.64 & 0.81 & 0.93 & 0.94 & 0.97 & 0.95 & 0.95 & 0.95 \\
$T=100$ & 0.72 & 0.86 & 0.94 & 0.97 & 0.97 & 0.96 & 0.96 & 0.97 \\
$T=250$ & 0.80 & 0.92 & 0.94 & 0.96 & 0.95 & 0.96 & 0.95 & 0.96 \\
$T=500$ & 0.90 & 0.94 & 0.96 & 0.95 & 0.96 & 0.96 & 0.96 & 0.94 \\
$T=1000$ & 0.92 & 0.95 & 0.95 & 0.95 & 0.96 & 0.95 & 0.95 & 0.95 \\
\hline 
Length\\
$T=20$ & 3.95 & 3.91 & 3.70 & 3.40 & 4.89 & 4.84 & 4.72 & 4.01 \\
$T=50$ & 2.32 & 2.30 & 2.18 & 2.14 & 2.96 & 2.90 & 2.63 & 2.24 \\
$T=100$ & 1.69 & 1.64 & 1.52 & 1.51 & 2.18 & 2.08 & 1.72 & 1.54 \\
$T=250$ & 1.05 & 0.98 & 0.95 & 0.95 & 1.34 & 1.16 & 0.99 & 0.96 \\
$T=500$ & 0.71 & 0.68 & 0.67 & 0.67 & 0.90 & 0.74 & 0.68 & 0.68 \\
$T=1000$ & 0.49 & 0.48 & 0.48 & 0.48 & 0.59 & 0.49 & 0.48 & 0.48 \\
\hline
MSE for $\delta$ \\
$T=20$ & 3.95 & 3.63 & 2.20 & 0.87 & 1.16 & 1.18 & 1.10 & 0.92 \\
$T=50$ & 1.35 & 1.07 & 0.47 & 0.32 & 0.44 & 0.49 & 0.41 & 0.32 \\
$T=100$ & 0.70 & 0.48 & 0.17 & 0.14 & 0.23 & 0.24 & 0.18 & 0.14 \\
$T=250$ & 0.24 & 0.09 & 0.06 & 0.06 & 0.11 & 0.08 & 0.07 & 0.06 \\
$T=500$ & 0.07 & 0.04 & 0.03 & 0.03 & 0.04 & 0.03 & 0.03 & 0.03 \\
$T=1000$ & 0.02 & 0.01 & 0.02 & 0.02 & 0.02 & 0.02 & 0.02 & 0.02 \\
\bottomrule
\end{tabular}          
}        
\caption{$\tau_0=0.3$}
\end{subtable}


\caption[Caption for delta1]{Simulation results for $\delta$}

\label{table_simulation_delta}
\end{table}
\FloatBarrier

Table \ref{table_simulation_tau} shows the results of estimation and inference of the break location $\tau$. Although the main focus of the current paper is on  inference about the slope parameters $\gamma$ and not on  inference about $\tau$, we report the empirical coverage and the length of the 95\% confidence interval of \cite{bai1997} and the highest posterior density (HPD) set\footnote{
Note that although we prove that the posterior mode of the break location $\tau$ converges to the limiting distribution of the least-squares estimator, whether the posterior distribution of $\tau$ converges to the same limit or not is still an open question. The Bayesian literature on the Bernstein-von Mises-like result for non-regular parameters is very scarce. To our best knowledge, the only available work is that by \cite{kleijn_knapik2012} whose results do not seem to be applicable to the model in consideration in this paper. Hence, it is not guaranteed that a credible set of $\tau$ has frequentist coverage even asymptotically. However, we emphasize that  credible sets on $\tau$ still have a statistically valid interpretation even in finite samples.}.
We also report the inverted likelihood ratio (ILR) confidence set suggested by \cite{eo_morley2015}. 

Overall, the HPD set and the ILR confidence set of the break location $\tau$ behave similarly
although the HPD set slightly undercovers relative to the ILR confidence set for small $T$ and/or $\delta_0$.
We  confirm several findings of Eo and Morley (2015). 
First, 
when $T$ is large, 
the confidence interval of Bai has longer lengths than the ILR confidence set and the HPD set\footnote{
\cite{eo_morley2015} explain that the likelihood ratio test is more powerful than the Wald-type test used to construct the confidence interval of Bai,
which results in a shorter length of the ILR confidence set.
}.
Second, when $T$ and $\delta_0$ are small,  the confidence interval of Bai tends to severely undercover compared to the ILR confidence set and the HPD set.
The interval of Bai indeed has a shorter length than the other two sets for small $T$, but its undercoverage raises concerns  for  small samples in practice\footnote{
Eo and Morley (2015) also find that 
 the confidence interval of \cite{qu_perron2007}  for the break location, which is also based on the Wald-type test as the confidence interval of Bai,
tends to undercover in small sample despite having a slightly shorter length than the ILR confidence set.
}
\footnote{
In addition, as also reported by Eo and Morley (2015), 
the ILR confidence set tends to slightly overcover even in large sample. 
}.

The bottom panels of Table \ref{table_simulation_tau} shows the mean-absolute-error (MAE) of the point estimator of $\tau$ which is $\hat{\tau}_{LS}$ defined in \eqref{tau_hat} for the conventional method and the posterior mode $\hat{\tau}_{Bayes}$ for the Bayesian approach. It is known that the finite-sample distribution of the least-squares estimator $\hat{\tau}_{LS}$ tends to be trimodal (see \citealp{baek2021}) when the jump size is relatively small. The same seems to be true for the Bayesian point estimator (see Figure \ref{figure_posterior_tau}).

%
\FloatBarrier
\begin{table}[h]
\begin{subtable}[h]{1\textwidth}
\centering
\scalebox{0.5}{        
\begin{tabular}{lllllllllllllllll}
    \toprule
    \toprule
    \multirow{1}{*}{} &
      \multicolumn{4}{c}{Least-squares} &
      \multicolumn{4}{c}{Bayesian} &
      \multicolumn{4}{c}{ILR} \\      
      \cmidrule(lr){2-5}\cmidrule(lr){6-9}\cmidrule(lr){10-13}      
      $\delta_0=$ & {0.25} & {0.50} & {1.00} &{2.00} & {0.25} & {0.50} & {1.00} & {2.00}& {0.25} & {0.50} & {1.00} & {2.00} \\
      \midrule
      Coverage\\
$T=20$ & 0.50 & 0.58 & 0.75 & 0.93 & 0.83 & 0.87 & 0.94 & 0.97 & 0.91 & 0.92 & 0.92 & 0.95 \\
$T=50$ & 0.51 & 0.68 & 0.87 & 0.97 & 0.83 & 0.92 & 0.96 & 0.97 & 0.93 & 0.93 & 0.96 & 0.98 \\
$T=100$ & 0.53 & 0.78 & 0.91 & 0.96 & 0.85 & 0.96 & 0.95 & 0.97 & 0.93 & 0.96 & 0.96 & 0.98 \\
$T=250$ & 0.67 & 0.87 & 0.94 & 0.97 & 0.91 & 0.94 & 0.94 & 0.97 & 0.94 & 0.95 & 0.96 & 0.98 \\
$T=500$ & 0.75 & 0.93 & 0.96 & 0.98 & 0.93 & 0.95 & 0.94 & 0.96 & 0.95 & 0.96 & 0.97 & 0.98 \\
$T=1000$ & 0.85 & 0.92 & 0.96 & 0.97 & 0.92 & 0.90 & 0.91 & 0.95 & 0.95 & 0.96 & 0.97 & 0.98 \\
\hline 
Length \small{($\times 100$)} \\
$T=20$ & 47.9 & 50.05 & 50.1 & 32.12 & 83.15 & 80.56 & 68.19 & 29.88 & 83.59 & 80.91 & 62.89 & 22.96 \\
$T=50$ & 48.3 & 50.66 & 40.26 & 14.0 & 76.18 & 69.37 & 40.09 & 9.01 & 80.92 & 71.41 & 37.08 & 9.18 \\
$T=100$ & 46.82 & 47.52 & 24.88 & 6.67 & 73.84 & 58.69 & 19.24 & 3.92 & 78.91 & 58.91 & 18.06 & 4.23 \\
$T=250$ & 47.63 & 34.35 & 9.78 & 2.66 & 64.2 & 30.63 & 5.83 & 1.50 & 67.73 & 28.93 & 6.37 & 1.68 \\
$T=500$ & 44.08 & 18.96 & 4.76 & 1.31 & 50.15 & 12.47 & 2.60 & 0.73 & 51.23 & 12.77 & 3.03 & 0.81 \\
$T=1000$ & 33.84 & 9.49 & 2.35 & 0.64 & 28.47 & 5.02 & 1.21 & 0.38 & 28.9 & 5.99 & 1.51 & 0.42 \\
\hline
MAE for $\tau$ \small{($\times 10$)} \\
$T=20$ & 2.63 & 2.27 & 1.43 & 0.37 & 3.13 & 2.72 & 1.64 & 0.38 \\
$T=50$ & 2.45 & 1.84 & 0.69 & 0.14 & 2.88 & 2.12 & 0.78 & 0.14 \\
$T=100$ & 2.32 & 1.24 & 0.38 & 0.06 & 2.76 & 1.50 & 0.39 & 0.06 \\
$T=250$ & 1.70 & 0.60 & 0.13 & 0.03 & 2.02 & 0.69 & 0.13 & 0.03 \\
$T=500$ & 1.20 & 0.27 & 0.06 & 0.01 & 1.31 & 0.28 & 0.06 & 0.01 \\
$T=1000$ & 0.63 & 0.13 & 0.03 & 0.01 & 0.70 & 0.14 & 0.03 & 0.01 \\
\bottomrule
\end{tabular}          
}
\caption{$\tau_0=0.5$}
\end{subtable}
    \hfill \\ 
    \quad \\
\begin{subtable}[h]{1\textwidth}
\centering
\scalebox{0.5}{        
\begin{tabular}{lllllllllllllllll}
    \toprule
    \toprule
    \multirow{1}{*}{} &
      \multicolumn{4}{c}{Least-squares} &
      \multicolumn{4}{c}{Bayesian} &
      \multicolumn{4}{c}{ILR} \\      
      \cmidrule(lr){2-5}\cmidrule(lr){6-9}\cmidrule(lr){10-13}      
      $\delta_0=$ & {0.25} & {0.50} & {1.00} &{2.00} & {0.25} & {0.50} & {1.00} & {2.00}& {0.25} & {0.50} & {1.00} & {2.00} \\
      \midrule
      Coverage\\
$T=20$ & 0.49 & 0.56 & 0.76 & 0.94 & 0.88 & 0.90 & 0.95 & 0.97 & 0.90 & 0.91 & 0.92 & 0.95 \\
$T=50$ & 0.54 & 0.68 & 0.87 & 0.96 & 0.89 & 0.94 & 0.96 & 0.98 & 0.94 & 0.93 & 0.95 & 0.98 \\
$T=100$ & 0.56 & 0.72 & 0.91 & 0.98 & 0.89 & 0.94 & 0.96 & 0.96 & 0.95 & 0.93 & 0.97 & 0.98 \\
$T=250$ & 0.65 & 0.88 & 0.94 & 0.97 & 0.92 & 0.96 & 0.95 & 0.97 & 0.93 & 0.95 & 0.97 & 0.99 \\
$T=500$ & 0.78 & 0.90 & 0.96 & 0.98 & 0.92 & 0.93 & 0.93 & 0.96 & 0.94 & 0.95 & 0.96 & 0.98 \\
$T=1000$ & 0.85 & 0.94 & 0.96 & 0.97 & 0.90 & 0.91 & 0.91 & 0.95 & 0.93 & 0.96 & 0.97 & 0.98 \\
\hline 
Length \small{($\times 100$)} \\
$T=20$ & 47.17 & 47.58 & 51.13 & 35.2 & 83.16 & 81.22 & 71.03 & 34.04 & 83.83 & 81.16 & 66.59 & 26.55 \\
$T=50$ & 49.07 & 48.74 & 40.61 & 15.03 & 76.84 & 70.0 & 43.7 & 9.62 & 81.42 & 72.56 & 40.55 & 9.62 \\
$T=100$ & 46.47 & 43.1 & 27.07 & 6.99 & 74.33 & 59.13 & 21.83 & 3.94 & 79.98 & 59.24 & 20.37 & 4.31 \\
$T=250$ & 45.62 & 34.95 & 10.59 & 2.67 & 65.97 & 33.65 & 5.87 & 1.52 & 70.08 & 31.54 & 6.24 & 1.71 \\
$T=500$ & 42.61 & 21.02 & 4.87 & 1.32 & 52.67 & 13.73 & 2.65 & 0.75 & 54.46 & 13.61 & 3.07 & 0.84 \\
$T=1000$ & 32.86 & 10.01 & 2.35 & 0.65 & 29.99 & 5.21 & 1.23 & 0.37 & 30.17 & 6.10 & 1.50 & 0.41 \\
\hline
MAE for $\tau$ \small{($\times 10$)} \\
$T=20$ & 3.00 & 2.64 & 1.64 & 0.39 & 3.43 & 3.03 & 1.86 & 0.41 \\
$T=50$ & 2.80 & 1.98 & 0.84 & 0.15 & 3.11 & 2.21 & 0.88 & 0.15 \\
$T=100$ & 2.46 & 1.49 & 0.37 & 0.06 & 2.86 & 1.73 & 0.41 & 0.06 \\
$T=250$ & 1.93 & 0.65 & 0.13 & 0.03 & 2.17 & 0.71 & 0.13 & 0.03 \\
$T=500$ & 1.26 & 0.32 & 0.06 & 0.01 & 1.39 & 0.32 & 0.06 & 0.01 \\
$T=1000$ & 0.67 & 0.13 & 0.03 & 0.01 & 0.68 & 0.14 & 0.03 & 0.01 \\
\bottomrule
\end{tabular}          
}
\caption{$\tau_0=0.3$}
\end{subtable}

\caption[]{
Simulation results for $\tau$
}
\label{table_simulation_tau}
\end{table}
\FloatBarrier

\clearpage 
To better understand the importance of the uncertainty of the break location $\tau$ for inference on the slope parameters, we conduct a hypothetical experiment.
We repeat the simulation exercise but now fixing  $\tau$ at the least-squares estimate $\hat{\tau}_{LS}$. Table \ref{table_simulation_delta_tauLS} displays the results. 
Note that the results for the least-squares estimator are of course the same as in Table  \ref{table_simulation_delta}. 
We however now see that, not only the conventional confidence intervals of $\delta$ but also the credible intervals  undercover for small $T$ and/or small $\delta_0$. They also have similar lengths in general.  Importantly, the credible intervals  when $\tau$ is fixed at $\hat{\tau}_{LS}$ (Table \ref{table_simulation_delta_tauLS}) have shorter lengths compared to the full Bayesian intervals  (Table \ref{table_simulation_delta}). On average, the full Bayesian credible intervals are 17.1\% longer\footnote{The difference is larger when $T$ and/or $\delta_0$ are/is smaller.} than the credible intervals produced by fixing the value of $\tau$ at $\hat{\tau}_{LS}$. 
Note that a Bayesian equivalent of the conventional approach to inference on the slope parameters  would be to fix the value of $\tau$ at the posterior mode (whose value is very similar to $\hat{\tau}_{LS}$ as we can see from Figure \ref{figure_posterior_tau} and deduce from Corollary \ref{ThmMode}). 
We can see in Figure \ref{figure_posterior_tau} that both $\hat{\tau}_{LS}$ and the posterior mode of $\tau$ display significant amount of variations. 
Fixing $\tau$ at a point estimate forces the Bayesian approach to ignore this  uncertainty of $\tau$; as a result, the credible interval on $\delta$ becomes shorter and hence undercovers.
The full Bayesian approach  takes into account such uncertainty via marginal posterior of $\tau$ (see examples of the density in Figure \ref{figure_posterior_tau}). This results in longer lengths of the full Bayesian intervals on 
the slope parameters and helps them avoid  undercoverage. In contrast, by construction (i.e.,\ Equation \ref{conventional_theory}), the conventional confidence intervals do not have this feature.

%
\FloatBarrier
\begin{table}[h]
\begin{subtable}[h]{0.5\textwidth}
\centering
\scalebox{0.55}{        
\begin{tabular}{lllllllllllll}
    \toprule
    \toprule
    \multirow{1}{*}{} &
      \multicolumn{4}{c}{Least-squares} &
      \multicolumn{4}{c}{Bayesian} \\
      \cmidrule(lr){2-5}\cmidrule(lr){6-9}
      $\delta_0=$ & {0.25} & {0.50} & {1.00} &{2.00} & {0.25} & {0.50} & {1.00} & {2.00} \\
      \midrule
      Coverage\\
$T=20$ & 0.68 & 0.77 & 0.88 & 0.95 & 0.65 & 0.70 & 0.88 & 0.95 \\
$T=50$ & 0.67 & 0.84 & 0.94 & 0.96 & 0.64 & 0.82 & 0.93 & 0.94 \\
$T=100$ & 0.69 & 0.90 & 0.96 & 0.95  & 0.69 & 0.89 & 0.95 & 0.95 \\
$T=250$ & 0.83 & 0.94 & 0.94 & 0.96& 0.83 & 0.94 & 0.95 & 0.94 \\
$T=500$ & 0.91 & 0.95 & 0.95 & 0.96& 0.88 & 0.96 & 0.95 & 0.95 \\
$T=1000$ & 0.93 & 0.94 & 0.95 & 0.96 & 0.94 & 0.95 & 0.95 & 0.96 \\
\hline 
Length\\
$T=20$ & 3.87 & 3.60 & 3.20 & 2.82 & 3.46 & 3.30 & 2.99 & 2.76 \\
$T=50$ & 2.31 & 2.07 & 1.82 & 1.76  & 2.21 & 2.04 & 1.79 & 1.74 \\
$T=100$ & 1.61 & 1.38 & 1.26 & 1.24& 1.56 & 1.36 & 1.23 & 1.24 \\
$T=250$ & 0.93 & 0.81 & 0.78 & 0.78 & 0.93 & 0.80 & 0.78 & 0.78 \\
$T=500$ & 0.61 & 0.56 & 0.55 & 0.55  & 0.61 & 0.56 & 0.55 & 0.55 \\
$T=1000$ & 0.41 & 0.39 & 0.39 & 0.39  & 0.41 & 0.39 & 0.39 & 0.39 \\
\hline
MSE for $\delta$\\
$T=20$ & 3.85 & 2.79 & 1.75 & 0.60 & 3.00 & 2.59 & 1.39 & 0.47 \\
$T=50$ & 1.35 & 0.78 & 0.26 & 0.20  & 1.28 & 0.79 & 0.28 & 0.20 \\
$T=100$ & 0.67 & 0.28 & 0.11 & 0.10 & 0.65 & 0.27 & 0.10 & 0.10 \\
$T=250$ & 0.18 & 0.05 & 0.04 & 0.04& 0.18 & 0.05 & 0.04 & 0.04 \\
$T=500$ & 0.05 & 0.02 & 0.02 & 0.02 & 0.06 & 0.02 & 0.02 & 0.02 \\
$T=1000$ & 0.02 & 0.01 & 0.01 & 0.01  & 0.01 & 0.01 & 0.01 & 0.01 \\
\bottomrule
\end{tabular}          
}
\caption{$\tau_0=0.5$}
\end{subtable}
\begin{subtable}[h]{0.5\textwidth}
\centering
\scalebox{0.55}{        
\begin{tabular}{lllllllllllll}
    \toprule
    \toprule
    \multirow{1}{*}{} &
      \multicolumn{4}{c}{Least-squares} &
      \multicolumn{4}{c}{Bayesian} \\
      \cmidrule(lr){2-5}\cmidrule(lr){6-9}
      $\delta_0=$ & {0.25} & {0.50} & {1.00} &{2.00} & {0.25} & {0.50} & {1.00} & {2.00} \\
      \midrule
      Coverage\\
$T=20$ & 0.66 & 0.70 & 0.87 & 0.93  & 0.65 & 0.71 & 0.87 & 0.92 \\
$T=50$ & 0.64 & 0.81 & 0.93 & 0.94  & 0.61 & 0.78 & 0.92 & 0.95 \\
$T=100$ & 0.72 & 0.86 & 0.94 & 0.97  & 0.67 & 0.87 & 0.93 & 0.95 \\
$T=250$ & 0.80 & 0.92 & 0.94 & 0.96& 0.81 & 0.93 & 0.95 & 0.96 \\
$T=500$ & 0.90 & 0.94 & 0.96 & 0.95 & 0.88 & 0.94 & 0.96 & 0.94 \\
$T=1000$ & 0.92 & 0.95 & 0.95 & 0.95  & 0.92 & 0.95 & 0.94 & 0.94 \\
\hline 
Length\\
$T=20$ & 3.95 & 3.91 & 3.70 & 3.40 & 3.49 & 3.53 & 3.48 & 3.27 \\
$T=50$ & 2.32 & 2.30 & 2.18 & 2.14 & 2.28 & 2.24 & 2.14 & 2.10 \\
$T=100$ & 1.69 & 1.64 & 1.52 & 1.51  & 1.66 & 1.57 & 1.49 & 1.49 \\
$T=250$ & 1.05 & 0.98 & 0.95 & 0.95 & 1.02 & 0.97 & 0.95 & 0.95 \\
$T=500$ & 0.71 & 0.68 & 0.67 & 0.67 & 0.71 & 0.68 & 0.67 & 0.67 \\
$T=1000$ & 0.49 & 0.48 & 0.48 & 0.48  & 0.49 & 0.48 & 0.48 & 0.48 \\
\hline
MSE for $\delta$ \\
$T=20$ & 3.95 & 3.63 & 2.20 & 0.87 & 3.03 & 2.83 & 1.76 & 0.76 \\
$T=50$ & 1.35 & 1.07 & 0.47 & 0.32& 1.34 & 1.01 & 0.47 & 0.29 \\
$T=100$ & 0.70 & 0.48 & 0.17 & 0.14  & 0.72 & 0.40 & 0.17 & 0.14 \\
$T=250$ & 0.24 & 0.09 & 0.06 & 0.06 & 0.20 & 0.09 & 0.06 & 0.06 \\
$T=500$ & 0.07 & 0.04 & 0.03 & 0.03 & 0.08 & 0.03 & 0.03 & 0.03 \\
$T=1000$ & 0.02 & 0.01 & 0.02 & 0.02  & 0.02 & 0.02 & 0.02 & 0.02 \\
\bottomrule
\end{tabular}          
}        
\caption{$\tau_0=0.3$}
\end{subtable}


\caption{
Simulation results for $\delta$, $\tau$ fixed at $\hat{\tau}_{LS}$
}
\label{table_simulation_delta_tauLS}
\end{table}
\FloatBarrier

\pagebreak

In summary, the simulation exercises demonstrate that 
(1) the credible intervals on the slope coefficient tend to have more reasonable coverages than the conventional confidence intervals  because of longer lengths,
(2) the longer length of the credible intervals is a reflection of the uncertainty of the unknown\footnote{When $\tau_0$ is known, the two intervals behave very similarly. To illustrate this point, we conduct another hypothetical experiment by repeating the simulation exercise as before but now fixing the value of $\tau$ at the true value $\tau_0$ in both conventional and Bayesian approaches. Table \ref{table_simulation_delta_tau0} summarizes the results. In this case, we see that both confidence and credible intervals have coverages quite close to 95\% in all cases. They also have similar lengths. Note that when the true value $\tau_0$ is given, the usual asymptotic normality and the regular Bernstein-von Mises theorem apply. As a consequence, both  frequentist and Bayesian intervals seem to converge faster to the limit compared to the case with unknown $\tau$.} break location $\tau$, and 
(3) the two intervals converge to each other asymptotically as expected from our Bernstein-von Mises theorem. 



\FloatBarrier
\begin{table}[!htb]
\begin{subtable}[h]{0.5\textwidth}
\centering
\scalebox{0.55}{        
\begin{tabular}{lllllllllllll}
    \toprule
    \toprule
    \multirow{1}{*}{} &
      \multicolumn{4}{c}{Least-squares} &
      \multicolumn{4}{c}{Bayesian} \\
      \cmidrule(lr){2-5}\cmidrule(lr){6-9}
      $\delta_0=$ & {0.25} & {0.50} & {1.00} &{2.00} & {0.25} & {0.50} & {1.00} & {2.00} \\
      \midrule
      Coverage\\
$T=20$ & 0.93 & 0.95 & 0.96 & 0.95 & 0.94 & 0.95 & 0.95 & 0.94 \\
$T=50$ & 0.95 & 0.95 & 0.94 & 0.94 & 0.94 & 0.94 & 0.94 & 0.94 \\
$T=100$ & 0.96 & 0.96 & 0.96 & 0.96 & 0.96 & 0.96 & 0.95 & 0.96 \\
$T=250$ & 0.95 & 0.95 & 0.95 & 0.94 & 0.94 & 0.96 & 0.94 & 0.94 \\
$T=500$ & 0.96 & 0.95 & 0.95 & 0.96 & 0.96 & 0.95 & 0.95 & 0.96 \\
$T=1000$ & 0.95 & 0.96 & 0.96 & 0.95 & 0.95 & 0.96 & 0.96 & 0.95 \\
\hline 
Length\\
$T=20$ & 2.77 & 2.77 & 2.77 & 2.77 & 2.72 & 2.74 & 2.73 & 2.74 \\
$T=50$ & 1.75 & 1.75 & 1.75 & 1.75 & 1.75 & 1.74 & 1.74 & 1.75 \\
$T=100$ & 1.24 & 1.24 & 1.24 & 1.24 & 1.23 & 1.24 & 1.24 & 1.23 \\
$T=250$ & 0.78 & 0.78 & 0.78 & 0.78 & 0.78 & 0.78 & 0.78 & 0.78 \\
$T=500$ & 0.55 & 0.55 & 0.55 & 0.55 & 0.55 & 0.55 & 0.55 & 0.55 \\
$T=1000$ & 0.39 & 0.39 & 0.39 & 0.39 & 0.39 & 0.39 & 0.39 & 0.39 \\
\hline
MSE for $\delta$\\
$T=20$ & 0.53 & 0.50 & 0.47 & 0.52 & 0.51 & 0.47 & 0.45 & 0.50 \\
$T=50$ & 0.21 & 0.20 & 0.21 & 0.21 & 0.21 & 0.20 & 0.20 & 0.20 \\
$T=100$ & 0.10 & 0.09 & 0.10 & 0.09 & 0.10 & 0.09 & 0.10 & 0.09 \\
$T=250$ & 0.04 & 0.04 & 0.04 & 0.04 & 0.04 & 0.04 & 0.04 & 0.04 \\
$T=500$ & 0.02 & 0.02 & 0.02 & 0.02 & 0.02 & 0.02 & 0.02 & 0.02 \\
$T=1000$ & 0.01 & 0.01 & 0.01 & 0.01 & 0.01 & 0.01 & 0.01 & 0.01 \\
\bottomrule
\end{tabular}          
}
\caption{$\tau_0=0.5$}
\end{subtable}
\begin{subtable}[h]{0.5\textwidth}
\centering
\scalebox{0.55}{        
\begin{tabular}{lllllllllllll}
    \toprule
    \toprule
    \multirow{1}{*}{} &
      \multicolumn{4}{c}{Least-squares} &
      \multicolumn{4}{c}{Bayesian} \\
      \cmidrule(lr){2-5}\cmidrule(lr){6-9}
      $\delta_0=$ & {0.25} & {0.50} & {1.00} &{2.00} & {0.25} & {0.50} & {1.00} & {2.00} \\
      \midrule
      Coverage\\
$T=20$ & 0.96 & 0.95 & 0.94 & 0.96 & 0.95 & 0.94 & 0.94 & 0.94 \\
$T=50$ & 0.95 & 0.94 & 0.95 & 0.94 & 0.95 & 0.94 & 0.94 & 0.94 \\
$T=100$ & 0.94 & 0.95 & 0.96 & 0.95 & 0.94 & 0.94 & 0.96 & 0.95 \\
$T=250$ & 0.95 & 0.96 & 0.94 & 0.95 & 0.95 & 0.95 & 0.94 & 0.96 \\
$T=500$ & 0.96 & 0.96 & 0.94 & 0.96 & 0.96 & 0.96 & 0.94 & 0.96 \\
$T=1000$ & 0.96 & 0.95 & 0.96 & 0.95 & 0.96 & 0.95 & 0.96 & 0.95 \\
\hline 
Length\\
$T=20$ & 3.37 & 3.37 & 3.37 & 3.37 & 3.30 & 3.31 & 3.31 & 3.34 \\
$T=50$ & 2.13 & 2.13 & 2.13 & 2.13 & 2.12 & 2.12 & 2.12 & 2.12 \\
$T=100$ & 1.51 & 1.51 & 1.51 & 1.51 & 1.50 & 1.50 & 1.50 & 1.51 \\
$T=250$ & 0.95 & 0.95 & 0.95 & 0.95 & 0.95 & 0.95 & 0.95 & 0.95 \\
$T=500$ & 0.67 & 0.67 & 0.67 & 0.67 & 0.67 & 0.67 & 0.67 & 0.67 \\
$T=1000$ & 0.48 & 0.48 & 0.48 & 0.48 & 0.48 & 0.48 & 0.48 & 0.48 \\
\hline
MSE for $\delta$ \\
$T=20$ & 0.73 & 0.75 & 0.74 & 0.73 & 0.68 & 0.70 & 0.69 & 0.69 \\
$T=50$ & 0.30 & 0.31 & 0.31 & 0.31 & 0.29 & 0.30 & 0.30 & 0.30 \\
$T=100$ & 0.16 & 0.16 & 0.14 & 0.16 & 0.15 & 0.15 & 0.14 & 0.15 \\
$T=250$ & 0.06 & 0.06 & 0.06 & 0.06 & 0.06 & 0.06 & 0.06 & 0.06 \\
$T=500$ & 0.03 & 0.03 & 0.03 & 0.03 & 0.03 & 0.03 & 0.03 & 0.03 \\
$T=1000$ & 0.01 & 0.01 & 0.01 & 0.01 & 0.01 & 0.01 & 0.01 & 0.01 \\
\bottomrule
\end{tabular}          
}        
\caption{$\tau_0=0.3$}
\end{subtable}


  \caption[Caption for LOF]{Simulation results for $\delta$, $\tau$ fixed at $\tau_0$}
\label{table_simulation_delta_tau0}
\end{table}
\FloatBarrier


\pagebreak
\section{Application}\label{section-application}

In this section, we illustrate difference in estimation and inference of the regression parameters in linear regression models with a structural break between the conventional approach and the Bayesian approach that we consider in this paper. 
\cite{paye_timmermann2006} consider the problem of ex-post prediction in stock returns under a structural break in the coefficients of state variables. 
Their multivariate model with a structural break is 
\[
Ret_{t}
=
 \begin{cases}
 \delta^{(1)}_1 + \delta^{(2)}_1 Div_{t-1} + \delta^{(3)}_1 Tbill_{t-1} + \delta^{(4)}_1 Spread_{t-1} + \delta^{(5)}_1 Def_{t-1} + \epsilon_t , & \text{ if } t \leq \floor{\tau T}  \\
 \delta^{(1)}_2 + \delta^{(2)}_2 Div_{t-1} + \delta^{(3)}_2 Tbill_{t-1} + \delta^{(4)}_2 Spread_{t-1} + \delta^{(5)}_2 Def_{t-1} + \epsilon_t , & \text{ if } t > \floor{\tau T} ,
\end{cases} 
\]
where 
$Ret_t$ is the excess return for the international index in question during month $t$, 
$Div_{t-1}$ is the lagged dividend yield,
$Tbill_{t-1}$ is the lagged local country short interest rate, 
$Spread_{t-1}$ is the lagged local country term spread, and 
$Def_{t-1}$ is the lagged U.S. default premium.    
The authors estimate the model using the conventional frequentist approach: 
they first compute $\hat{\tau}_{LS}$ and then obtain point estimates as well as confidence intervals for the slope coefficients by fixing $\tau$ at $\hat{\tau}_{LS}$.
We examine whether the Bayesian method performs differently from the conventional approach.

Monthly series are collected from Global Financial Data and Federal Reserve Economic Data (FRED).
In this paper, we consider estimating the model for the United Kingdom and Japan\footnote{
\cite{paye_timmermann2006} conduct the sequential method suggested by \cite{bai_perron1998}, \cite{bai_perron2003}, \cite{perron2006} for determining the number of breaks
and find multiple breaks for some countries. They find single breaks for the U.K. and Japan, but, for example, two breaks for the U.S.
A fully Bayesian approach would be to place a prior on the number of breaks and use a trans-dimensional estimation method such as a reversible jump MCMC, which is beyond the scope of this paper.
}.
The indices to which the total return and the dividend yield correspond are the FTSE All-share for the U.K. and Nikko Securities Composite for Japan. 
For each country, a 3-month Treasury bill rate is used as a measure of the short interest rate while the yield on a long-term government bond is used as a measure of the long interest rate. 
Excess returns are computed as the total return on stocks in the local currency minus the local short rate. 
The dividend yield is expressed as an annual rate and is constructed as the sum of dividends over the preceding 12 months, divided by the current price. 
A term spread is the difference between the long and short local country interest rates.
The U.S. default premium is defined as the difference in yields between Moody's Baa and Aaa rated bonds. 
For each country, the sample spans between January 1970 and December 2003. 
%


For both approaches, we set the range of the candidate values of $\tau$ to be $(\epsilon,1-\epsilon)$ with $\epsilon=0.05$ as we do in the simulation studies in the previous section.
For the Bayesian approach, 
we use the uniform prior on $(\epsilon,1-\epsilon)$ for $\tau$ and the conjugate prior for the regression parameters with $H=0.1 I_{(d_x+d_z)}$, $\underline{\mu}=0_{(d_x+d_z)}$, and $\underline{a}=\underline{b}=1$. The findings are similar even when we use the uninformative improper prior. 
For the break date, we compute the least-squares estimator  $\hat{\tau}_{LS}$ and the posterior mode $\hat{\tau}_{Bayes}$  of $\tau$ as well as 
the 95\% confidence interval of \cite{bai1997}, 
the highest posterior density (HPD) set, and  
the inverted likelihood ratio (ILR) confidence set of \cite{eo_morley2015}.
For the slope parameters, we compute $\hat{\gamma}_{LS}$ and the posterior mean of $\gamma$ as well as 
the 90\% confidence intervals of \cite{bai1997} based on the asymptotic result \eqref{conventional_theory}
and 
the equal-tailed credible intervals.

\begin{table}
\centering
\resizebox{0.4\columnwidth}{!}{%

  \begin{threeparttable}
\begin{tabular}{ccccccccccc}
    \toprule
    \toprule
    \multirow{2}{*}{slopes} &
      \multicolumn{3}{c}{Least-squares} &
      \multicolumn{3}{c}{Bayesian} \\
      \cmidrule(lr){2-4}\cmidrule(lr){5-7}
                & {Estimate} & {LB} & {UB} &{Estimate} & {LB} & {UB}  \\
\midrule
$\delta_1^{(1)}$ & -21.2 &-28.1  & -14.2 & -18.6& -25.2& -11.9 \\
$\delta_1^{(2)}$ & -0.35 & -3.01 & 2.30 & 0.07  & -2.51  &2.64 \\
$\delta_1^{(3)}$ & -0.77 & -1.52 & -0.03 &-0.96 & -1.70&-0.23 \\
$\delta_1^{(4)}$ &0.80 &-0.57 & 2.19 & 0.80 & -0.56&2.14 \\
$\delta_1^{(5)}$ &19.4  &11.8 & 27.0 &16.5 &9.18&  23.8 \\
\midrule
$\delta_2^{(1)}$ & 19.1&11.9 &26.4& 16.5  & 9.62& 23.3\\
$\delta_2^{(2)}$ & 1.42 &-1.43 & 4.29 &0.99  & -1.77 & 3.77 \\
$\delta_2^{(3)}$ &-0.36 &  -1.20 & 0.47 & -0.16 & -0.99 &0.64 \\
$\delta_2^{(4)}$ &-0.98 &-2.42& 0.45 & -0.98  & -2.39  &0.41 \\
$\delta_2^{(5)}$ & -19.4 &  -27.2& -11.6 & -16.5  &-23.9 & -8.95 \\
\midrule
\midrule
    \multirow{2}{*}{$\tau$} &
      \multicolumn{3}{c}{Least-squares} &
      \multicolumn{3}{c}{Bayesian} \\
      \cmidrule(lr){2-4}\cmidrule(lr){5-7}
                & {Estimate} & {LB} & {UB} &{Estimate} & {LB} & {UB}  \\
\midrule                
\multirow{2}{*}{} 
 & 0.150 &0.145 &0.155 & 0.150 & 0.149& 0.152 \\
 & (75:01) &(74:11) & (75:03) & (75:01) & (74:12)& (75:02)\\
 
     \multirow{2}{*}{} &
      \multicolumn{3}{c}{ILR} &
      \multicolumn{3}{c}{} \\
      \cmidrule(lr){2-4}
                & {Estimate} & {LB} & {UB} &&  &   \\
\midrule                
\multirow{2}{*}{} 
 &  &0.149 & 0.151 &  & &  \\
 &  &(74:12) & (75:02) &  & &  \\

\bottomrule
\bottomrule
\end{tabular} 
    \begin{tablenotes}
      \small
      \item The upper panel shows
point estimates as well as 90\% confidence (left) and equal-tailed credible (right) intervals for the regression slope parameters. 
The lower panel shows
point estimates of $\tau$ with the corresponding months in parentheses as well as the bounds of 95\% confidence intervals of \cite{bai1997} and highest posterior density (HPD) sets. 
It also displays the inverted likelihood ratio (ILR) confidence sets of \cite{eo_morley2015}.
LB=lower bound and UB=upper bound of the intervals. 
    \end{tablenotes}
  \end{threeparttable}
  }
\caption{Estimation results for the U.K. stock return}  
\label{table_application_uk}  
\end{table}

When the uncertainty about $\tau$ is small, estimation and inference of the slope parameters roughly match between the conventional least-squares approach and the Bayesian approach, as illustrated by our simulation studies and indicated by our proven Bernstein-von mises theorem. 
See Table \ref{table_application_uk} for the results for the U.K. 
Both methods estimate a break at 1975:01.
The confidence interval of \cite{bai1997}, the Bayesian highest posterior density (HPD) set, and the inverted likelihood ratio (ILR) confidence set by \cite{eo_morley2015}  are all similar and narrow, indicating that the  uncertainty about $\tau$ is small. 
This can be seen also from the posterior density on the break date in Panel (a) of Figure \ref{fig_application_tau}, which has a sharp peak around 1975:01\footnote{
The mean and the standard deviation of the excess return of the FTSE All-share index during the sample period are -1.53 and 6.94 respectively.
%
At $t=$1974:12, we have $Ret_t=-9.9$ while at $t=$1975:01, $Ret_t=43.75$, where the change is approximately 7.7 standard deviations.
Therefore, the change in the dependent variable is large enough for the break point to be detected with small uncertainty.}.
\cite{paye_timmermann2006} explain that the break in the mid-1970's might be related to the  large macroeconomic shocks reflecting  oil price increases. 
As a result of the small uncertainty about $\tau$, 
the point estimates of the slope parameters as well as the corresponding confidence/credible intervals are similar between the conventional  and the Bayesian approach. 
Importantly, when the confidence interval of a given slope parameter  includes (or does not include) zero, the corresponding credible interval also  includes (or does not include) zero.
Hence, the conventional approach to inference about the slope parameters for the U.K. sample seems to be robust with respect to the uncertainty on the break date.

\FloatBarrier
\graphicspath{{Figures_final/application_oct2021/}}
\begin{figure}[ht] 
\centering
  \begin{subfigure}[b]{0.4\linewidth}
    \centering
    \includegraphics[width=\linewidth]{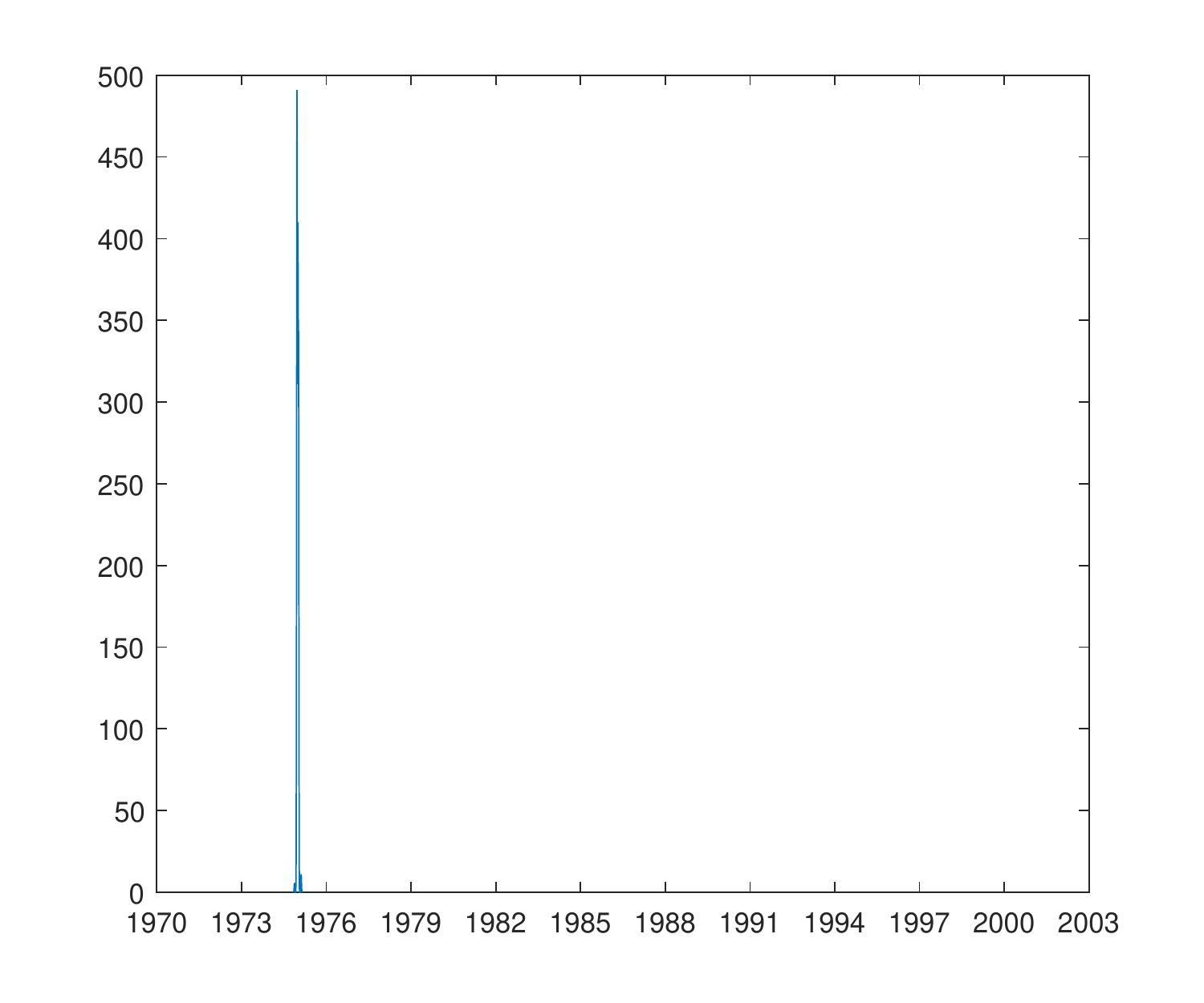} 
    \caption{United Kingdom} 
  \end{subfigure}
  \begin{subfigure}[b]{0.4\linewidth}
    \centering
    \includegraphics[width=\linewidth]{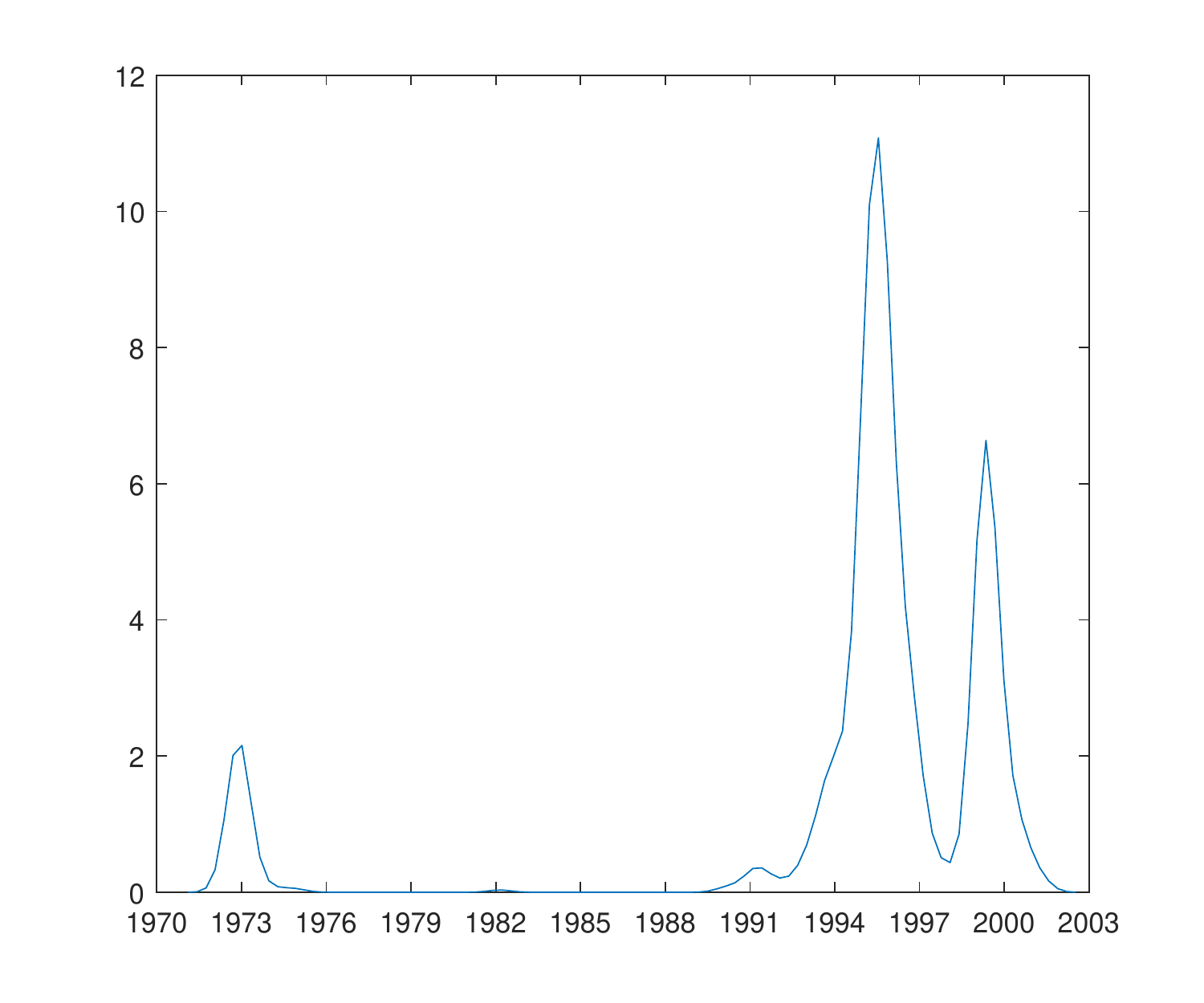} 
    \caption{Japan} 
  \end{subfigure}
\caption{\small{Posterior density of the break date}}\label{fig_application_tau}
\end{figure}
\FloatBarrier

\begin{table}
\centering
\resizebox{0.4\columnwidth}{!}{%

  \begin{threeparttable}
\begin{tabular}{ccccccccccc}
    \toprule
    \toprule
    \multirow{2}{*}{slopes} &
      \multicolumn{3}{c}{Least-squares} &
      \multicolumn{3}{c}{Bayesian} \\
      \cmidrule(lr){2-4}\cmidrule(lr){5-7}
                & {Estimate} & {LB} & {UB} &{Estimate} & {LB} & {UB}  \\
\midrule
$\delta_1^{(1)}$ & 2.41 &0.24  & 4.57 & 1.49  &-1.89& 5.27 \\
$\delta_1^{(2)}$ & 1.17 &0.44 & 1.89 &1.31 & 0.29 & 3.26 \\
$\delta_1^{(3)}$ & -1.98 & -2.45 & -1.50 & -2.00  &-3.53 &-1.23 \\
$\delta_1^{(4)}$ &-0.85 &-1.55    & -0.16& -0.29 &-1.43 &3.59\\
$\delta_1^{(5)}$ & 2.76 & 1.39 & 4.12 &2.29 &0.15  & 4.11 \\
\midrule
$\delta_2^{(1)}$ &  -16.8  &-29.6 &  -4.07 & -9.39  & -23.4& 8.77 \\
$\delta_2^{(2)}$ & 12.4 & 2.35 & 22.4 &6.22 & -5.91 & 16.8 \\
$\delta_2^{(3)}$ &  -7.11 & -14.1 & -0.04 & -3.86 & -13.1 &5.32 \\
$\delta_2^{(4)}$ & 4.76 &1.25 & 8.27 & 1.70& -5.86 &6.59 \\
$\delta_2^{(5)}$ & -8.35 & -14.9 & -1.78 &-4.65&-12.1 & 4.98 \\
\midrule
\midrule
    \multirow{2}{*}{$\tau$} &
      \multicolumn{3}{c}{Least-squares} &
      \multicolumn{3}{c}{Bayesian} \\
      \cmidrule(lr){2-4}\cmidrule(lr){5-7}
                & {Estimate} & {LB} & {UB} &{Estimate} & {LB} & {UB}  \\
\midrule                
\multirow{2}{*}{} 
 & 0.780 &0.777 & 0.782 & 0.779 & 0.068& 0.934 \\
 & (96:05) &(96:04) & (96:06) & (96:05) & (72:03)&(01:10) \\
 
     \multirow{2}{*}{} &
      \multicolumn{3}{c}{ILR} &
      \multicolumn{3}{c}{} \\
      \cmidrule(lr){2-4}
                & {Estimate} & {LB} & {UB} &&  &   \\
\midrule                
\multirow{2}{*}{} 
 &  &0.080 &0.890 &  & &  \\
 &  &(72:08) & (00:02) &  & &  \\

\bottomrule
\bottomrule
\end{tabular}  
    \begin{tablenotes}
      \small
      \item The upper panel shows
point estimates as well as 90\% confidence (left) and equal-tailed credible (right) intervals for the regression slope parameters. 
The lower panel shows
point estimates of $\tau$ with the corresponding months in parentheses as well as the bounds of 95\% confidence intervals of \cite{bai1997} and highest posterior density (HPD) sets. 
It also displays the inverted likelihood ratio (ILR) confidence sets of \cite{eo_morley2015}.
LB=lower bound and UB=upper bound of the intervals. 
    \end{tablenotes}
  \end{threeparttable}
  }
\caption{Estimation results for the Japanese stock return}  
\label{table_application_jpn}  
\end{table}

In contrast, when the uncertainty on $\tau$ is large, the conventional and the Bayesian results on inference about the slope parameters might disagree.
Table \ref{table_application_jpn} shows the results for Japan. 
Although both $\hat{\tau}_{LS}$ and the posterior mode of $\tau$ are at 1996:05, 
the HPD set and the ILR confidence set are much wider than the confidence interval of \cite{bai1997}, indicating a large uncertainty of the break date. 
The posterior density on $\tau$ in Figure \ref{fig_application_tau} also illustrates that the uncertainty of the break date is much larger for Japan than for the U.K. during the sample period\footnote{
In addition, the posterior on $\tau$ for Japan exhibits tri-modality, which would be similar to the tendency of a finite-sample distribution of $\hat{\tau}_{LS}$ to have three modes as reported in the literature (e.g.,\ \citealp{baek2021}; \citealp{casini_perron2021continuous}). 
}.
The  large uncertainty of $\tau$ is reflected on  Bayesian inference on the slope parameters. 
In the upper panel of Table \ref{table_application_jpn},
we see that in general the Bayesian credible intervals are wider than the confidence intervals. 
Importantly, this can have a qualitative consequence on statistical importance of some parameters. 
For seven of the ten slope coefficients, the confidence intervals do not include zero while the the Bayesian credible intervals do. 
Hence, 
the conventional approach to inference on the slope parameters might not be robust 
with respect to 
the uncertainty of the break date, for the Japanese sample.


\section{Conclusion and future direction}\label{section-conclusion}
In this paper, we establish a Bernstein-von Mises type theorem for the slope coefficients in linear regression  with a structural break. By doing so, we bridge the gap between the frequentist and the Bayesian approaches for inference on this  model. On the one hand, a frequentist researcher can look at Bayesian credible intervals for the slope coefficients as a robustness check to see whether the uncertainty of the break location affects  inference on the slope parameters. Such sensitivity analysis is natural as our theoretical result guarantees the credible interval to converge to the conventional confidence interval that the frequentist researcher would use otherwise. On the other hand,  Bayesian inference can be conveyed to frequentists via our proven result. 



Potential extensions include several directions.
First, the homoscedasticity assumption could be too strong in some applications, and hence extending the results to the case of heteroscedasticity and autocorrelation would be of interest. 
Second, a popular Bayesian method of \cite{chib1998} is different from the approach we took in this paper in that we place an explicit prior on $\tau$ and that Chib's framework can be naturally extended to the case of multiple breaks.
It would be interesting to study frequentist properties of Chib's approach.

\pagebreak
\appendix
\section{Proof of Theorems and Corollaries}\label{AppendixThms}
In Appendix A, we provide proofs of Theorems 1-4 and Corollaries 1-2. See Appendix \ref{AppendixProps} for proofs of the Propositions used for proving the main theorems. 

\subsection{Proof of Theorem 1}\label{ProofThmConsistency}

\begin{proof}[Proof of Theorem \ref{ThmConsistency}]
%
Note that 
\begin{align*}
\pi_T(\tau) &= \frac{ L_T( \tau) }{ \int L_T( \tau')  d\tau'  } = \frac{ L_T( \tau_0) }{ \int L_T( \tau')  d\tau'  } \frac{ L_T( \tau)   }{ L_T( \tau_0)  } =\pi_T(\tau_0)  \frac{ L_T(\tau)   }{ L_T(\tau_0)  } ,\\
\pi_T(\tau_0) &=  \frac{ L_T( \tau_0) }{ \int L_T( \tau')  d\tau'  } \leq  \frac{ L_T(\tau_0) }{ \int _{   B^c_{M_0/T}(\tau_0) }L_T( \tau')  d\tau'  }  
=\bigg[ \int_{  B^c_{M_0/T} (\tau_0)  } \frac{ L_T( \tau') }{ L_T( \tau_0 )} d\tau'   \bigg]^{-1},
\end{align*}
for any $M_0>0$.
Hence for each $T$ and for any $M_0>0$, 
\begin{equation}
\int_{ B^c_{M/T}(\tau_0) }  \pi_T(\tau) d\tau 
=
\pi_T(\tau_0)   \int_{ B^c_{M/T}(\tau_0) } \frac{ L_T( \tau)   }{ L_T(\tau_0)  } d\tau 
\leq 
\bigg[ \int_{  B^c_{M_0/T} (\tau_0) } \frac{ L_T(\tau') }{ L_T( \tau_0 )} d\tau'   \bigg]^{-1}  
\int_{ B^c_{M/T}(\tau_0) } \frac{ L_T( \tau)   }{ L_T(\tau_0)  } d\tau.  \label{bounding}
\end{equation}
Therefore, we want to find 
\begin{enumerate}
\item
an upper bound for $\int_{ B^c_{M/T}(\tau_0) } \frac{ L_T( \tau)   }{ L_T(\tau_0)  } d\tau $ and 
\item 
a lower bound for $ \int_{  B^c_{M_0/T} (\tau_0) } \frac{ L_T(\tau') }{ L_T( \tau_0 )} d\tau'  $ for some $M_0>0$
\end{enumerate}
We can write the marginal likelihood ratio as   
\[
 \frac{  L_T(\tau)  }{  L_T(\tau_0) } 
=
\exp\bigg[ T\left\{  \frac{1}{T} \log \left( \frac{L_T(\tau)}{L_T(\tau_0)} \right)   \right\} \bigg] .
\]

The proof of Theorem 1 is built on some intermediate steps, Propositions \ref{PropRatio}-\ref{PropACts}.
Proposition \ref{PropRatio} shows that, under the normal likelihood and the conjugate prior, studying this ratio boils down to comparing the sum of squared residuals  $S_T(\tau)$. 

\begin{restatable}{proposition}{PropRatio}\label{PropRatio}%
Suppose  Assumption \ref{assumption1} holds. 
Then, with the normal likelihood 
and the conjugate prior described above,  
under $P_{\theta_0, \tau_0}$, for all $\tau$, 
\[
 \frac{1}{T} \log \left( \frac{L_T(\tau)}{L_T(\tau_0)} \right) 
 =
\frac{1}{2} \log \left( \frac{S_T(\tau_0)}{S_T(\tau)} \right)
+
O_p(T^{-1}).
\]
\end{restatable}
Let us first examine the limit of the quantity $Q_T(\tau) = T^{-1}S_T(\tau)$.
Proposition \ref{PropLimitQ} states that $Q_T(\tau)$ converges in probability to some deterministic function $Q(\tau)$.
See Figure \ref{fig_Qn} for examples of $Q_T(\tau)$ and $Q(\tau)$.
\begin{restatable}{proposition}{PropLimitQ}\label{PropLimitQ}%
Suppose  Assumption \ref{assumption1} holds. Then, under $P_{\theta_0, \tau_0}$, for all $\tau$, 
\[
Q_T(\tau) = Q(\tau) + O_p(T^{-1/2}),
\]
where 
\[
Q(\tau) 
=
\sigma_0^2
+
\begin{cases}
(\tau_0-\tau)
\frac{(1-\tau_0)   }{(1-\tau) }
\delta_0'R'
\Sigma_X
R\delta_0 ,& \text{if } \tau \leq \tau_0\\
  (\tau-\tau_0)
\frac{\tau_0   }{\tau}
\text{     }
\delta_0' R'
\Sigma_X
R\delta_0  
,& \text{if } \tau > \tau_0
\end{cases}
\equiv 
\sigma_0^2+\Delta(\tau) .
\]
\end{restatable}
Define $A_T(\tau) = g\left(  Q_T(\tau)   \right)$ and  $A(\tau) = g(Q(\tau))$ where $g(x)=-\frac{1}{2}\log(x)$. 
Due to Proposition \ref{PropRatio}, we can write
\begin{equation}
 T^{-1} \log \left( \frac{L_T(\tau)}{L_T(\tau_0)} \right) 
=
A_T(\tau) - A_T(\tau_0) +O_p(T^{-1}). \label{ratio}
\end{equation}

Proposition \ref{PropAMax} below says that the limit $A(\tau)$ of $A_T(\tau)$ attains its maximum at $\tau_0$.
\begin{restatable}{proposition}{PropAMax}\label{PropAMax}%
$A(\tau)$ attains its unique maximum at $\tau_0$
\end{restatable}
%
Proposition \ref{PropACts} establishes the modulus of continuity of the empirical process $\left\{  A_T(\tau) -A_T(\tau_0) \right\} -   \left\{  A(\tau) -A(\tau_0)    \right\}$ outside of a ball around $\tau_0$ with radius proportional to $T^{-1}$.
\begin{restatable}{proposition}{PropACts}\label{PropACts}
Suppose  Assumption \ref{assumption1} holds. Then, under
$\forall \eta>0$, $\forall \epsilon>0$, $\exists M>0$ and $k>0$ such that $T \geq k$ $\implies $
\[
P_{\theta_0, \tau_0} 
\left( 
\inf_{\tau \in B^c_{M/T}(\tau_0)} \frac{| \left\{  A_T(\tau) -A_T(\tau_0) \right\} -   \left\{  A(\tau) -A(\tau_0)    \right\}   |}{| \tau-\tau_0| } < \eta
\right)
> 1-\epsilon.
\]
\end{restatable}


By Proposition \ref{PropAMax}, $A(\cdot)$ attains its unique max at $\tau_0$. Note that the convex function $A(\tau)$ is not differentiable at $\tau_0$. Hence we have, 
\begin{align*}
A(\tau)-A(\tau_0) &< |\tau - \tau_0 |  B_1,  \\ 
A(\tau)-A(\tau_0) &> |\tau - \tau_0 | B_2 ,  
\end{align*}
for some $B_1,B_2<0$.
%
By Proposition \ref{PropACts}, given $\eta_1>0$, $\exists M>0:$ with $P_{\theta_0, \tau_0} \to 1$,
\begin{equation}
A_T(\tau) -A_T(\tau_0) 
< 
\eta_1|\tau-\tau_0| +A(\tau) -A(\tau_0)
< 
|\tau-\tau_0| \{   \eta_1 +B_1  \},  \label{ub}
\end{equation}
for all $\tau \in B^c_{M/T}(\tau_0)$. 
Similarly, given $\eta_2>0$, $\exists M_0>0:$ with $P_{\theta_0, \tau_0}\to1$,
\begin{equation}
A_T(\tau) -A_T(\tau_0) 
>
-\eta_2 |\tau-\tau_0| +A(\tau) -A(\tau_0) 
> 
|\tau-\tau_0| \{  - \eta_2 +B_2  \},  \label{lb}
\end{equation}
for all $\tau \in B^c_{M_0/T}(\tau_0)$.
Recall, by Eq.\eqref{ratio}, we have
\[
\frac{L_T(\tau)}{L_T(\tau_0)} 
=
\exp\bigg[ T \bigg( A_T(\tau_0)-A_T(\tau)  \bigg) +O_p(1) \bigg].
\]
Hence, from Eq. \eqref{ub}, given $\eta_1>0$, small compared to $-B_1$, there is $B_1'<0$, which is independent of $M$: we have with $P_{\theta_0, \tau_0}\to1$,
\begin{equation}
\frac{L_T(\tau)}{L_T(\tau_0)} 
\leq
\exp\bigg[ T|\tau-\tau_0|  B'_1    +O_p(1) \bigg] 
=
\exp\bigg[ T|\tau-\tau_0|  B'_1  \bigg]
O_p(1),
\label{exp_ub}
\end{equation}
for all $\tau \in B^c_{M/T}(\tau_0)$. Note that the statement above still holds with a larger value of $M>0$ as the area outside of the ball will be contained by that for the original $M$.
Similarly, from Eq. \eqref{lb}, there is $B_2'<0$ and $M_0>0:$  with $P_{\theta_0, \tau_0}\to1$,
\begin{equation}
\frac{L_T(\tau)}{L_T(\tau_0)} 
\geq
\exp\bigg[ T|\tau-\tau_0|  B'_2   +O_p(1) \bigg]
=
\exp\bigg[ T|\tau-\tau_0|  B'_2  \bigg]
O_p(1) ,
\label{exp_lb}
\end{equation}
for all $\tau \in B^c_{M_0/T}(\tau_0)$.
Now, by Inequality \eqref{exp_ub} and the fundamental theorem of calculus, 
\[
\int_{B^c_{M/T}(\tau_0)} \frac{L_T(\tau)}{L_T(\tau_0)} d\tau
\leq 
\int_{B^c_{M/T}(\tau_0)} \exp\bigg[ T|\tau-\tau_0|  B'_1  \bigg]d\tau O_p(1) 
=
\frac{1}{TB'_1} \left(  e^{TB'_1} - e^{B'_1M} \right)
O_p(1) .
\]
Similarly, by Inequality \eqref{exp_lb}, 
\[
\int_{B_{M_0^c/T}(\tau_0)} \frac{L_T(\tau)}{L_T(\tau_0)} d\tau
\geq 
\int_{B_{M_0^c/T}(\tau_0)} \exp\bigg[ T|\tau-\tau_0|  B'_2   \bigg] d\tau O_p(1) 
=
\frac{1}{TB'_2} \left(  e^{TB'_2} - e^{B'_2 M_0} \right)
O_p(1) .
\]
This means, together with the bound \eqref{bounding}, 
\[
\int_{ B^c_{M/T}(\tau_0) }  \pi_T(\tau) d\tau 
\leq 
\bigg[ \int_{  B_{M^c_0/T} (\tau_0) } \frac{ L_T(\tau') }{ L_T( \tau_0 )} d\tau'   \bigg]^{-1}  
\int_{ B^c_{M/T}(\tau_0) } \frac{ L_T( \tau)   }{ L_T(\tau_0)  } d\tau  
\leq 
\frac{B'_2}{B'_1} \frac{ e^{B'_1 T}-e^{B'_1M} }{ e^{B'_2 T} - e^{B_2' M_0}  }
O_p(1) ,
\]
which can be made arbitrarily small by choosing $M>0$ and $T$ sufficiently large.
\end{proof}

\subsection{Proof of Corollary 1}\label{ProofThmMode}
\begin{proof}[Proof of Corollary \ref{ThmMode} ]
The main structure of the proof follows Proposition 2 of \cite{bai1997} and relies on an implication of our Theorem \ref{ThmConsistency}. First, note that we have 
\begin{align*}
\hat{\tau}_{Bayes} 
&=
\argmax_{\tau \in \mathcal{H}} \pi_T(\tau) \\
&=
\argmax_{\tau \in \mathcal{H}} L_T(\tau) \\
&=
\argmax_{\tau \in \mathcal{H}} \frac{1}{T} \log \left( \frac{L_T(\tau)}{L_T(\tau_0)} \right) ,
\end{align*}
which converges in distribution to $\argmax_{\tau \in \mathcal{H}} \log \left( \frac{S_T(\tau_0)}{S_T(\tau)} \right)$ by Proposition \ref{PropRatio}. We have 
\begin{align*}
\argmax_{\tau \in \mathcal{H}} \log \left( \frac{S_T(\tau_0)}{S_T(\tau)} \right)
&=
\argmin_{\tau \in \mathcal{H}}S_T(\tau)\\
&=
\argmax_{\tau \in \mathcal{H}}V_T(\tau)-V_T(\tau_0),
\end{align*}
where $V_T(\tau) = \hat{\delta}(\tau)' \left( Z_2'MZ_2 \right) \hat{\delta}(\tau)$.  \cite{bai1997} shows that $V_T(\tau)-V_T(\tau_0)$ converges in distribution to $W^*\left(\floor*{ T(\tau- \tau_0) } \right)$ uniformly on any bounded interval around $\tau_0$. Let $m^* = \argmax_m W^*(m)$, which is $O_p(1)$. Hence, $\forall \epsilon>0$, $\exists R_1>0:$ $P\left( |m^*| > R_1 \right) <\epsilon$. Our Theorem \ref{ThmConsistency} implies that $\hat{\tau}_{Bayes} = \tau_0 + O_p(T^{-1})$. In other words, $\forall \epsilon>0$, $\exists R_2>0:$ $P\left( T|\hat{\tau}_{Bayes} - \tau_0 |> R_2 \right) <\epsilon$. Take $R=\max\{ R_1,R_2\}$.

Define $\hat{\tau}_R =\argmax_{T|\tau - \tau_0 | \leq R } V_T(\tau) - V_T(\tau_0)$ and $m^*_R =\argmax_{|m | \leq R } W^*(m)$. Then we have $T|\hat{\tau}_R - \tau_0| \overset{d}{\to} m^*_R$. In other words, $\big| P\left( \floor*{T(\hat{\tau}_R -\tau_0 )}=j \right) - P\left( m^*_R = j \right) \big| < \epsilon$ as $T \to \infty$ $\forall |j| \leq R$.  

Note that if $T|\hat{\tau}_{Bayes} - \tau_0 |<R$, then $\hat{\tau}_R=\hat{\tau}_{Bayes}$. Similarly, if $|m^*| <R$, then $m^*_R=m^*$. Hence, $\big| P\left( \floor*{T(\hat{\tau}_{Bayes} -\tau_0 )}=j \right) - P\left( m^* = j \right) \big| $ is bounded by $\big| P\left( \floor*{T(\hat{\tau}_R -\tau_0 )}=j \right) - P\left( m^*_R = j \right) \big|  + P\left(T|\hat{\tau}_{Bayes} - \tau_0|  \geq R\right)  + P\left( |m^*|  \geq R\right) < 3 \epsilon$. As $\epsilon$ can be made arbitrarily small, the desired result holds.

\end{proof}

\subsection{Proof of Theorem 2}\label{ProofThmBvm}
\begin{proof}[Proof of Theorem \ref{ThmBvm} ]
Define $z = \sqrt{T}\left( \gamma- \hat{\gamma}_{LS} \right)$ and let $\phi(x; \mu, \Sigma)$  be the multivariate normal density with mean $\mu$ and covariance matrix $\Sigma$ evaluated at $x$. 
\begin{align*}
&d_{TV} \bigg(  
\pi \left[ z |  \bm{D}_T\right], 
N_{(d_x+d_z)}\left(  0, \sigma_0^2 V^{-1} \right) 
\bigg)
=
\int |  \pi(z| \bm{D}_T) -\phi(z; 0, \sigma_0^2 V^{-1})  | dz\\
&\leq
\int \int  |  \pi(z|\tau,  \bm{D}_T) -\phi(z; 0, \sigma_0^2 V^{-1})  | dz d\pi(\tau| \bm{D}_T)\\
&=
\int d_{TV}\bigg( \pi(z|\tau, \bm{D}_T) , \phi(z; 0, \sigma_0^2 V^{-1})   \bigg) d\pi(\tau| \bm{D}_T)\\
&=
\int_{B_{M/T}(\tau_0)} d_{TV}\bigg( \pi(z|\tau,  \bm{D}_T) , \phi(z; 0, \sigma_0^2 V^{-1})  \bigg) d\pi(\tau| \bm{D}_T) +o_p(1),
\end{align*}
where the last equality is due to Theorem \ref{ThmConsistency}.

From \eqref{posterior_gamma_conjugate}, asymptotically, the posterior of $\gamma$ conditional on $\tau$ is normal:
\begin{align*}
\gamma
\big| \tau, \bm{D}_T &\overset{a}{\sim} N_{(d_x+d_z)} 
\left(
\bar{\mu}_{\tau},
( \bar{b}_\tau /  \bar{a} )\bar{H}_\tau^{-1}
\right) \\
\implies
z | \tau , \bm{D}_T
&\overset{a}{\sim} N_{(d_x+d_z)}
\left(
\sqrt{T}\left( \bar{\mu}_{\tau}-\hat{\gamma}_{LS} \right),
(T \bar{b}_\tau /  \bar{a} )\bar{H}_\tau^{-1}
\right).
\end{align*}
%

The total variation distance is bounded above by 2 times square root of the KL divergence. 
In general, the KL divergence between two $p$-dimensional normal distributions $N_p(\mu_1,\Sigma_1)$ and $N_p(\mu_2,\Sigma_2) $ is bounded above by 
\begin{equation}\label{KLbound}
\underbrace{
\frac{ \left|  \det\left( \Sigma_2^{-1}\right) - \det\left( \Sigma_1^{-1}\right)  \right| }{\min( \det\left( \Sigma_1^{-1}\right) ,\det\left( \Sigma_2^{-1}\right) )  } 
}_{I}
+ 
\underbrace{
p ||\Sigma_2^{-1} - \Sigma_1^{-1} ||_\infty ||\Sigma_1||_\infty
}_{II} 
+
\underbrace{
||\mu_1-\mu_2||_2^2 ||\Sigma_2^{-1}||_2
}_{III} , 
\end{equation}
where $||\Sigma||_\infty = max_{ij} |\Sigma_{ij}|$ is the largest element of $\Sigma$ in the absolute value, and $||\Sigma||_2 = sup_\mu ||\Sigma\mu||_2/||\mu||_2$ is a matrix norm induced by the standard norm on $\mathbb{R}^p$, $||\mu||_2 = \sum_{i=1}^p \mu_i^2$.
We can bound the total variation distance between the posterior density of $z$ conditional on $\tau$ and that of $N_{(d_x+d_z)}(0, \sigma_0^2 V^{-1})$ using the bound \eqref{KLbound}, with
$\mu_1 =  
\sqrt{T}\left( \bar{\mu}_{\tau}-\hat{\gamma}_{LS} \right)$,
$\Sigma_1=(T \bar{b}_\tau /  \bar{a} )\bar{H}_\tau^{-1}$, 
$\mu_2=0$, and
$\Sigma_2=  \sigma_0^2V^{-1}$. 


To show $III=o_p(1)$, we write
\begin{equation}
\sqrt{T}\left( \bar{\mu}_{\tau}-\hat{\gamma}_{LS} \right)
=
\sqrt{T}\left( \bar{\mu}_{\tau}-\hat{\gamma}(\tau) \right)
+
\sqrt{T}\left( \hat{\gamma}(\tau) -\hat{\gamma}_{LS} \right). \label{mean}
\end{equation}
By definition, 
\[
\bar{\mu}_{\tau} 
=
 \left[ \frac{1}{T} \underline{H} +\frac{1}{T}  \chi_\tau'\chi_\tau\right]^{-1} \left[ \frac{1}{T}  \underline{H}\underline{\mu}+\frac{1}{T} \chi_\tau Y   \right]  
=
\hat{\gamma}(\tau) + O_p(T^{-1}),
\]
so the first term in \eqref{mean} is  $o_p(1)$. 
To show that  the second term in \eqref{mean} is $o_p(1)$ for $\tau \in B_{M/T}(\tau_0)$, 
write $\sqrt{T}\left(  \hat{\gamma}(\tau) - \hat{\gamma}_{LS} \right)=\sqrt{T}\left(  \hat{\gamma}(\tau) - \gamma_0 \right)-\sqrt{T}\left(  \hat{\gamma}_{LS} -\gamma_0 \right)$. 
Note that $Y=X\beta_0 + Z_{2\tau_0}\delta_0 + \epsilon = X\beta_0+ Z_{2\tau}\delta_0+\epsilon_\tau^*$, where $\epsilon_\tau^* = ( Z_{2\tau_0} - Z_{2\tau } ) \delta_0 + \epsilon$. This implies 
\begin{align*}
\sqrt{T}\left(  \hat{\gamma}(\tau) - \gamma_0 \right)
&=
\left[ 
\frac{1}{T}
\begin{pmatrix}
X'X              & X'Z_{2\tau} \\
Z_{2\tau}'X &  Z_{2\tau}'Z_{2\tau}
\end{pmatrix}
\right]^{-1}
\frac{1}{ \sqrt{T} }
\begin{pmatrix}
X'\epsilon_\tau^*\\
Z_{2\tau}'\epsilon_\tau^*
\end{pmatrix}\\
&=
\left[ 
\frac{1}{T}
\begin{pmatrix}
X'X              & X'Z_{2\tau} \\
Z_{2\tau}'X &  Z_{2\tau}'Z_{2\tau}
\end{pmatrix}
\right]^{-1}
\frac{1}{ \sqrt{T} }
\begin{pmatrix}
X'\epsilon             +X'                 ( Z_{2\tau_0}-Z_{2\tau} ) \delta_0\\
Z_{2\tau}'\epsilon +Z_{2\tau_0}'(Z_{2\tau_0}-Z_{2\tau})  \delta_0 
\end{pmatrix}.
\end{align*}
For $|\tau - \tau_0| < \frac{M}{T}$, we have
\begin{align*}
\frac{1}{T}X'Z_{2\tau}-\frac{1}{T}X'Z_{2\tau_0} &=o_p(1), 
\quad
\frac{1}{T}Z_{2\tau}'Z_{2\tau}-\frac{1}{T}Z_{2\tau_0}'Z_{2\tau_0} =o_p(1), \\
\quad 
\frac{1}{ \sqrt{T} }X'                 ( Z_{2\tau_0}-Z_{2\tau} ) &=o_p(1), 
\quad  
\frac{1}{ \sqrt{T} } Z_{2\tau_0}'(Z_{2\tau_0}-Z_{2\tau}) =o_p(1), \\
\quad 
\frac{1}{ \sqrt{T} } Z_{2\tau}'\epsilon - \frac{1}{\sqrt{T} } Z_{2\tau_0}\epsilon &=o_p(1),
\end{align*}
which implies 
\[
\sqrt{T}\left(  \hat{\gamma}(\tau) -\gamma_0 \right)
=
\left[ 
\frac{1}{T}
\begin{pmatrix}
X'X              & X'Z_{2\tau_0} \\
Z_{2\tau_0}'X &  Z_{2\tau_0}'Z_{2\tau_0}
\end{pmatrix}
\right]^{-1}
\frac{1}{ \sqrt{T} }
\begin{pmatrix}
X'\epsilon            \\
Z_{2\tau_0}'\epsilon 
\end{pmatrix}
+o_p(1).
\]
Similarly, since the least-square estimator $\hat{\tau}_{LS} \in B_{M/T}(\tau_0)$ for sufficiently large $T$, we can show 
\begin{align*}
\sqrt{T}\left(  \hat{\gamma}_{LS}- \gamma_0 \right)
&=
\sqrt{T}\left(  \hat{\gamma}(\hat{\tau}_{LS}) - \gamma_0 \right)\\
&=
\left[ 
\frac{1}{T}
\begin{pmatrix}
X'X              & X'Z_{2\hat{\tau}_{LS}} \\
Z_{2\hat{\tau}_{LS}}'X &  Z_{2\hat{\tau}_{LS} }'Z_{2\hat{\tau}_{LS} }
\end{pmatrix}
\right]^{-1}
\frac{1}{ \sqrt{T} }
\begin{pmatrix}
X'\epsilon             +X'                 ( Z_{2\tau_0}-Z_{2\hat{\tau}_{LS} } ) \delta_0\\
Z_{2\hat{\tau}_{LS} }'\epsilon +Z_{2\tau_0}'(Z_{2\tau_0}-Z_{2\hat{\tau}_{LS} })  \delta_0 
\end{pmatrix}\\
&=
\left[ 
\frac{1}{T}
\begin{pmatrix}
X'X              & X'Z_{2\tau_0} \\
Z_{2\tau_0}'X &  Z_{2\tau_0}'Z_{2\tau_0}
\end{pmatrix}
\right]^{-1}
\frac{1}{ \sqrt{T} }
\begin{pmatrix}
X'\epsilon            \\
Z_{2\tau_0}'\epsilon 
\end{pmatrix}
+o_p(1).
\end{align*}
Hence, 
$\sqrt{T}\left(  \hat{\gamma}(\tau) - \hat{\gamma}_{LS} \right)=o_p(1)$. 

To show $I$ and $II$ are $o_p(1)$,  note that $\Sigma_1-\Sigma_2$ equals to 
\begin{align}
(T \bar{b}_\tau /  \bar{a} )\bar{H}_\tau^{-1} - \sigma^2_0V^{-1}
&=
\left[ (T \bar{b}_\tau /  \bar{a} )\bar{H}_\tau^{-1} - \frac{TS_T(\tau)}{T-(d_x+d_z)} ( \chi_\tau \chi_\tau )^{-1} \right]
+
\left[ \frac{TS_T(\tau)}{T-(d_x+d_z)} ( \chi_\tau \chi_\tau )^{-1}- \sigma^2_0V^{-1} \right]. \label{var}
\end{align}
For the first term in \eqref{var}, we have
\begin{align*}
(T \bar{b}_\tau /  \bar{a} )\bar{H}_\tau^{-1}
&=
\frac{ \bar{b}_\tau }{\underline{a} +T/2} \left[ \frac{1}{T} \underline{H} +\frac{1}{T}  \chi_\tau'\chi_\tau\right]^{-1} 
=
\frac{ \frac{1}{T}\bar{b}_\tau }{\underline{a}/T +1/2} \left[ \frac{1}{T} \underline{H} +\frac{1}{T}  \chi_\tau'\chi_\tau\right]^{-1} .
\end{align*}
Note that  $(1/T)\bar{b}_{\tau}=\frac{1}{2T} S_T(\tau) +O_p(T^{-1})$, 
so we have 
$
(T \bar{b}_\tau /  \bar{a} )\bar{H}_\tau^{-1}
=
S_T(\tau) 
\left[ \chi_\tau'\chi_\tau\right]^{-1}
+O_p(T^{-1}).
$
Therefore, the term in the first square brackets in \eqref{var} is $o_p(1)$. 
For the second term in \eqref{var}, we have that for $|\tau - \tau_0|<\frac{M}{T}$,
\begin{align*}
 \frac{TS_T(\tau)}{T-(d_x+d_z)} ( \chi_\tau \chi_\tau )^{-1}- \sigma^2_0V^{-1}
&=
\underbrace{
\left( 
Q_T(\tau)
-
Q_T(\tau_0)
\right)
}_{=o_p(1)}
\hat{V}^{-1} _T(\tau) 
+
Q_T(\tau_0)
\underbrace{
\left( 
\hat{V}^{-1} _T(\tau)
-
\hat{V}^{-1} _T(\tau_0)
\right)
}_{=o_p(1)}
\\
&+
\underbrace{
\left(
Q_T(\tau_0) 
-
\sigma_0^2 
 \right)
}_{=O_p(T^{-1/2})}
\hat{V}^{-1}_T(\tau_0)
+
\sigma_0^2
\underbrace{
\left(
\hat{V}^{-1}_T(\tau_0)
-
V^{-1}
\right)
}_{=o_p(1)}
+
o_p(1)
=o_p(1),
\end{align*}
where $\hat{V}_T(\tau)= \frac{1}{T}\chi_\tau' \chi_\tau$.

This implies that 
$\Sigma_2^{-1}-\Sigma_1^{-1} 
=o_p(1)
$. Hence $II =o_p(1)$. By continuity of determinants, we also have that $I=o_p(1)$ for $\tau \in B_{M/T}(\tau_0)$.

Finally, for $\tau \in B_{M/T}(\tau_0)$
\begin{align*}
d_{TV}\left( \pi(z|\tau, \bm{D}_T) , N_{(d_x+d_z)}(0, \sigma_0^2 V)  \right)
\leq
2 \sqrt{o_p(1)}
=o_p(1).
\end{align*}
This implies that  $d_{TV} \bigg(  
\pi \left[ z | \bm{D}_T\right], 
N_{(d_x+d_z)}(0, \sigma_0^2 V^{-1}) 
\bigg)$ is bounded above by
\[
\int_{B_{M/T}(\tau_0)} d_{TV}\bigg( \pi(z|\tau,  \bm{D}_T) , \phi(z; 0, \sigma_0^2 V^{-1})  \bigg) d\pi(\tau| \bm{D}_T) +o_p(1)=o_p(1).
\]

\end{proof}

\pagebreak

\subsection{Proof of Theorem 3}\label{ProofThmConsistencyGeneral}
\begin{proof}[Proof of Theorem \ref{ThmConsistencyGeneral} ]

Recall that the proof of Theorem \ref{ThmConsistency} is an implication of Propositions 1-4.
Assumption \ref{assumption1} implies Propositions 2-4.
Proposition \ref{PropRatio} establishes that under the normal likelihood and the conjugate prior, Assumption \ref{assumption1} implies
\begin{align}
 \frac{1}{T} \log \left( \frac{L_T(\tau)}{L_T(\tau_0)} \right) 
 &=
\frac{1}{2}
\log \left( \frac{S_T(\tau_0)}{S_T(\tau)} \right)
+
O_p(T^{-1}). \label{desired_eq}
\end{align}
Therefore, Theorem \ref{ThmConsistencyGeneral} can be proved if we establish the above equality under the normal likelihood and the prior described in Section \ref{section-general-theory}, together with Assumptions \ref{assumption1}-\ref{assumption2}.

For a given $\tau$, denote by $F_T(\theta,\tau)=\log p(\bm{Y}_T |\bm{X}_T, \theta, \tau)$ the log  likelihood function conditional on $\tau$.
Under the normal likelihood and Assumption \ref{assumption1}, together with Assumption \ref{assumption2}, 
we can involke 
Theorem 3 of \cite{hong_preston2012} (see their page 361) which establishes that
\begin{align*}
\log \int e^{ F_T(\theta, \tau) - F_T( \hat{\theta}(\tau) , \tau) } \pi(\theta, \tau) d\theta = \log\left[  \pi(\theta^*(\tau), \tau ) (2\pi)^{(d_x+d_z+1)/2} \det \left( -TA_\theta(\tau) \right)^{-1/2} \right] +o_p(1),
\end{align*}
for each $\tau$, where $-A_\theta(\tau) $ is the probability limit of 
$ -\frac{1}{T}\frac{\partial^2}{\partial \theta \partial \theta'} F_T\left( \hat{\theta}(\tau), \tau \right) $ and is positive definite.

Note that 
\begin{align*}
\frac{L_T(\tau)}{L_T(\tau_0)}
&=
\frac{ e^{F_T(\hat{\theta}(\tau), \tau)} \int e^{ F_T(\theta, \tau) - F_T( \hat{\theta}(\tau) , \tau) } \pi(\theta, \tau) d\theta }{ e^{F_T(\hat{\theta}(\tau_0), \tau_0)}\int e^{ F_T(\theta, \tau_0) - F_T( \hat{\theta}(\tau_0) , \tau_0) )} \pi(\theta, \tau_0) d\theta  },
\end{align*}
which implies that  
\begin{align*}
 \frac{1}{T} \log \left( \frac{L_T(\tau)}{L_T(\tau_0)} \right) 
&=
\frac{1}{T} F_T(\hat{\theta}(\tau), \tau)
-
\frac{1}{T} F_T(\hat{\theta}(\tau_0), \tau_0)\\
&+
\frac{1}{T}  \log\left[  \pi(\theta^*(\tau), \tau ) (2\pi)^{(d_x+d_z+1)/2} \det \left( -TA_\theta(\tau) \right)^{-1/2} \right] \\
&-
\frac{1}{T}  \log\left[  \pi(\theta^*(\tau_0), \tau_0 ) (2\pi)^{(d_x+d_z+1)/2} \det \left( -TA_\theta(\tau_0) \right)^{-1/2} \right] 
+
O_p(T^{-1})\\
&=
\frac{1}{T}\log \left[  \frac{p(\bm{Y}_T  |\bm{X}_T, \hat{\theta}(\tau),\tau)}{p(\bm{Y}_T  |\bm{X}_T, \hat{\theta}(\tau_0),\tau_0)} \right]
+
\frac{1}{T}\log \left[  \frac{\pi \left(  \theta^*(\tau), \tau \right) }{ \pi \left(  \theta^*(\tau_0), \tau_0 \right) } \right]\\
&-
\frac{1}{2T} \log \left[ \frac{\det \left( -A_\theta(\tau)  \right) }{ \det \left( -A_\theta(\tau_0) \right) } \right]
+
O_p(T^{-1}).
\end{align*} 
Note that we assumed that $ \pi \left(  \theta^*(\tau), \tau \right) $ and $ \pi \left(  \theta^*(\tau_0), \tau_0 \right) $ are finite and non-zero. Hence, the term involving the ratio of the priors is $O_p(T^{-1})$. 
Also,  $-A_\theta(\tau)$ is a positive definite matrix hence its determinant is a finite positive number. 
We have 
\[
p(\bm{Y}_T  |\bm{X}_T, \hat{\theta}(\tau),\tau)
\propto
\left( \frac{1}{\hat{\sigma}^2(\tau)} \right)^{T/2} \exp \bigg[ - \frac{1}{2\hat{\sigma}^2(\tau)} \underbrace{ \sum_{t=1}^T \left( y_t -\chi_{\tau,t} \hat{\gamma}(\tau)  \right)^2 }_{=S_T(\tau)} \bigg] 
=
\left( \frac{1}{\hat{\sigma}^2(\tau)} \right)^{T/2} \exp\left( -T/2\right), 
\]
where the last equality is due to the fact that $\hat{\sigma}^2(\tau) = S_T(\tau)/T$. 
This  implies the desired result  i.e.,\ \eqref{desired_eq}.
%
Note that Propositions 2-4 hold under Assumption \ref{assumption1}.
Therefore, given \eqref{desired_eq}, 
 the rest of the proof of Theorem \ref{ThmConsistencyGeneral} follows the same argument in the proof of Theorem \ref{ThmConsistency} in \ref{ProofThmConsistency}.
\end{proof}

\subsection{Proof of Corollary 2}\label{ProofThmModeGeneral}
\begin{proof}[Proof of Corollary \ref{ThmModeGeneral} ]
Note that 
\begin{align*}
\hat{\tau}_{Bayes} 
&=
\argmax_{\tau \in \mathcal{H}} \pi_T(\tau) 
=
\argmax_{\tau \in \mathcal{H}} \frac{1}{T} \log \left( \frac{L_T(\tau)}{L_T(\tau_0)} \right) ,
\end{align*}
 converges in distribution to $\argmax_{\tau \in \mathcal{H}} \log \left( \frac{S_T(\tau_0)}{S_T(\tau)} \right)$ due to \eqref{desired_eq}, 
 under the normal likelihood and the assumed condition on the prior, together with Assumptions \ref{assumption1}-\ref{assumption2}.
Furthermore, Theorem \ref{ThmConsistencyGeneral} implies that $\hat{\tau}_{Bayes} = \tau_0 + O_p(T^{-1})$.
Based on these two facts,  the rest of the proof follows the same argument as in the proof of Corollary \ref{ThmMode} in \ref{ProofThmMode}.

\end{proof}

\subsection{Proof of Theorem 4}\label{ProofThmBvmGeneral}
\begin{proof}[Proof of Theorem \ref{ThmBvmGeneral} ]

%
Under the normal likelihood and Assumptions \ref{assumption1}-\ref{assumption2}, a result from \cite{hong_preston2012} (see their page 367) implies that the posterior of $\sqrt{T}\left( \gamma-  \hat{\gamma}(\tau) \right)$ conditional on $\tau$ converges in total variation in probability to the multivariate normal distribution $N(0, -A_\gamma^{-1}(\tau) )$, where $A_\gamma^{-1}(\tau)$ is the sub-matrix of $A_\theta^{-1}(\tau)$ obtained by deleting the last row and the last column,  
$-A_\theta(\tau) $ is the probability limit of 
$ -\frac{1}{T}\frac{\partial^2}{\partial \theta \partial \theta'} F_T\left( \hat{\theta}(\tau), \tau \right) $, and 
$F_T(\theta,\tau)=\log p(\bm{Y}_T |\bm{X}_T, \theta, \tau)$ is the log  likelihood function conditional on $\tau$.
 This means that the total variation between the posterior of $z=\sqrt{T}\left( \gamma -  \hat{\gamma}_{LS} \right)$ given $\tau$ and $N\left( \sqrt{T}\left( \hat{\gamma}(\tau) -  \hat{\gamma}_{LS} \right), -A_\gamma^{-1}(\tau)  \right)$ converges to 0 in probability.

The bound \eqref{KLbound} on the KL divergence between two normal densities can be used again now with 
$\mu_2=0$,
$\Sigma_2=  \sigma_0^2V^{-1}$, 
$\mu_1 =  
 \sqrt{T}\left( \hat{\gamma}(\tau)- \hat{\gamma}_{LS} \right)$, and
$\Sigma_1 =  -A_\gamma^{-1}(\tau) = \text{plim } \hat{\sigma}^2(\tau) \hat{V}_T^{-1}(\tau)$ where $\hat{V}_T(\tau)= \frac{1}{T}\chi_\tau' \chi_\tau$. 
%
Note that from the proof of Theorem \ref{ThmBvm} in \ref{ProofThmBvm}, we know that $\mu_1=o_p(1)$ for $|\tau - \tau_0|<\frac{M}{T}$.

For $|\tau - \tau_0|<\frac{M}{T}$,
\begin{align*}
\Sigma_1-\Sigma_2
&=
\underbrace{
-A_\gamma^{-1}(\tau) -Q_T(\tau) \hat{V}_T^{-1}(\tau)
}_{=o_p(1)}
\\
&+
\underbrace{
\left( 
Q_T(\tau)
-
Q_T(\tau_0)
\right)
}_{=o_p(1)}
\hat{V}^{-1} _T(\tau) \\
&+
Q_T(\tau_0)
\underbrace{
\left( 
\hat{V}^{-1} _T(\tau)
-
\hat{V}^{-1} _T(\tau_0)
\right)
}_{=o_p(1)}
\\
&+
\underbrace{
\left(
Q_T(\tau_0) 
-
\sigma_0^2 
 \right)
}_{=O_p(T^{-1/2})}
\hat{V}^{-1}_T(\tau_0)
\\
&+
\sigma_0^2
\underbrace{
\left(
\hat{V}^{-1}_T(\tau_0)
-
V^{-1}
\right)
}_{=o_p(1)}
+
o_p(1)
=o_p(1),
\end{align*}
which implies 
\begin{align*}
\Sigma_2^{-1}-\Sigma_1^{-1} 
=o_p(1).
\end{align*}
The rest of the proof can be done similarly as in the proof of Theorem \ref{ThmBvm} in \ref{ProofThmBvm} by applying the bound \eqref{KLbound}.
\end{proof}

\pagebreak
\section{Proof of Propositions}\label{AppendixProps}


\subsection{Proof of Proposition 1}\label{ProofPropRatio}
\PropRatio*

\begin{proof}[Proof of Proposition \ref{PropRatio}]
From \eqref{posterior_tau_conjugate}, we have
\begin{align*}
  \frac{1}{T} \log \left( \frac{L_T(\tau)}{L_T(\tau_0)} \right)   
&=
\frac{1}{2T}\log \left[ 
\frac{
\det
\left( \bar{H}_{\tau_0}\right)
}
{
\det
\left(\bar{H}_{\tau}\right)
}
\right] 
+\frac{\bar{a}}{T} \log \bigg[ \frac{ \bar{b}_{\tau_0} }{ \bar{b}_{\tau} }\bigg] 
+\frac{1}{T}\log \bigg( \frac{\pi(\tau)}{\pi(\tau_0)} \bigg).
\end{align*}
Assumption \ref{assumption1}  implies that each component of $(1/T)\chi_\tau' \chi_\tau$ converges in probability to a constant matrix. By continuity of determinant, the determinant converges to the determinant of the limiting matrix. As a result, the quantity inside of log in the first term is $O_p(1)$ and hence the first term is $O_p(T^{-1})$.  
By the choice of the prior, the ratio $\pi(\tau)/\pi(\tau_0)$ is bounded, so the last term is $O(T^{-1})$. 
Note that  
\begin{align}
(1/T)\bar{b}_{\tau}
&=
(1/T)\underline{b}
+
\frac{1}{2T} \left[  \underline{\mu}' \underline{H} \underline{\mu} +Y'Y - \bar{\mu}_\tau' \bar{H}_\tau \bar{\mu}_\tau \right] \nonumber  \\
&=
\frac{1}{2T} \left[  Y'Y - (\chi_\tau'Y)' (\chi_\tau'\chi_\tau)^{-1} (\chi_\tau'Y) \right] +O_p(T^{-1}) \nonumber \\
&=
\frac{1}{2T} S_T(\tau) +O_p(T^{-1}) . \label{b_tau}
\end{align}
Hence, we conclude that
\begin{align*}
&  \frac{1}{T} \log \left( \frac{L_T(\tau)}{L_T(\tau_0)} \right)   
 =  \frac{1}{2} \log \bigg( \frac{S_T(\tau_0)}{S_T(\tau)} \bigg)+O_p(T^{-1}). 
\end{align*}




\end{proof}

\subsection{Proof of Proposition 2}\label{proof-prop2}
\PropLimitQ*

\begin{proof}[Proof of Proposition \ref{PropLimitQ}]
Let $\tau \in(0,1)$ be given. Let $M=I -X(X'X)^{-1}X'$. We have the following identity: $S_T(\tau) = \bar{S}_T  - V_T(\tau)$ (Amemiya, 1985; Bai, 1997), where $\bar{S}_T$ is the sum of squared residuals from regressing $Y$ on $X$ alone and $V_T(\tau) = \hat{\delta}'(\tau) ( Z_{2\tau}'M Z_{2\tau} ) \hat{\delta}(\tau)$. By Frisch-Waugh Theorem, the OLS estimate of $\delta$ in Eq. \eqref{model3} is equivalent to that in the model $MY = MZ_{2\tau}\delta+M\epsilon$. Note the true model is $MY = MZ_{2\tau_0}\delta_0+M\epsilon$. Hence, 
\begin{align*}
\hat{\delta}(\tau)
&=
(Z_{2\tau}'M Z_{2\tau})^{-1} Z_{2\tau}'M Y\\
&=
(Z_{2\tau}'M Z_{2\tau})^{-1} Z_{2\tau}'\left\{MZ_{2\tau_0}\delta_0+M\epsilon \right\} \\
&=
(Z_{2\tau}'M Z_{2\tau})^{-1}Z_{2\tau}'MZ_{2\tau_0}\delta_0+(Z_{2\tau}'M Z_{2\tau})^{-1}Z_{2\tau}'M\epsilon.
\end{align*}
We have 
\begin{align*}
V_T(\tau) 
&=
\delta_0' (Z_{2\tau}'MZ_{2\tau_0})'(Z_{2\tau}'M Z_{2\tau})^{-1}(Z_{2\tau}'MZ_{2\tau_0})\delta_0\\
&+
2\delta_0' (Z_{2\tau}'MZ_{2\tau_0})'(Z_{2\tau}'M Z_{2\tau})^{-1}(Z_{2\tau}'M\epsilon)\\
&+
(Z_{2\tau}'M\epsilon)'(Z_{2\tau}'M Z_{2\tau})^{-1}(Z_{2\tau}'M\epsilon).
\end{align*}
By Assumption \ref{assumption1} (i),(ii), and (iv),
\begin{align*}
\frac{1}{T}V_T(\tau) 
&=
\frac{1}{T}\delta_0' (Z_{2\tau}'MZ_{2\tau_0})'(Z_{2\tau}'M Z_{2\tau})^{-1}(Z_{2\tau}'MZ_{2\tau_0})\delta_0 + O_p(T^{-1/2}).
\end{align*}
Also we have 
\begin{align*}
\bar{S}_T
&=
Y'MY 
=\delta_0'Z_{2\tau_0}'M Z_{2\tau_0}\delta_0
+
2\delta_0Z_{2\tau_0}M
+
\epsilon'M\epsilon,
\end{align*}
which implies 
\begin{align*}
\frac{1}{T}\bar{S}_T
&=
\frac{1}{T}\delta_0'Z_{2\tau_0}'M Z_{2\tau_0}\delta_0
+\sigma_0^2
+O_p(T^{-1/2}).
\end{align*}
By the above identity, $Q_T(\tau) 
=
\frac{1}{T}\bar{S}_T
-
\frac{1}{T}V_T(\tau)
$ equals to
\[
\sigma_0^2
+
\frac{1}{T}  \delta_0'\bigg\{ (Z_{2\tau_0}'M Z_{2\tau_0})   - (Z_{2\tau}'MZ_{2\tau_0})'(Z_{2\tau}'M Z_{2\tau})^{-1}(Z_{2\tau}'MZ_{2\tau_0}) \bigg\}\delta_0
+O_p(T^{-1/2}).
\]
Note that 
\begin{align*}
(Z_{2\tau_0}'M Z_{2\tau_0}) 
&=
Z_{2\tau_0}'Z_{2\tau_0} -Z_{2\tau_0}'X(XX')^{-1} X' Z_{2\tau_0}\\
&=
R'X_{2\tau_0}' X_{2\tau_0}R -R' (X_{2\tau_0}'X)(XX')^{-1} (X' X_{2\tau_0})R\\
&=
R'X_{2\tau_0}' X_{2\tau_0}R -R' (X_{2\tau_0}'X_{2\tau_0})(XX')^{-1} (X_{2\tau_0}' X_{2\tau_0})R.
\end{align*}
By Assumption \ref{assumption1} (iii) and (iv),
\begin{align*}
\frac{1}{T}(Z_{2\tau_0}'M Z_{2\tau_0}) 
&=(1-\tau_0)R'\Sigma_XR -(1-\tau_0)^2R'\Sigma_XR +O_p(T^{-1/2})\\
&=
\tau_0(1-\tau_0) R'\Sigma_XR +O_p(T^{-1/2}).
\end{align*}
Similarly,
\begin{align*}
\frac{1}{T} (Z_{2\tau}'M Z_{2\tau}) 
&=
\tau(1-\tau) R'\Sigma_XR +O_p(T^{-1/2}).
\end{align*}
WLOG, suppose $\tau<\tau_0$. Then
\begin{align*}
(Z_{2\tau}'M Z_{2\tau_0}) 
&=
Z_{2\tau}'Z_{2\tau_0} -Z_{2\tau}'X(XX')^{-1} X' Z_{2\tau_0}\\
&=
R'X_{2\tau}' X_{2\tau_0}R -R' (X_{2\tau}'X)(XX')^{-1} (X' X_{2\tau_0})R\\
&=
R'X_{2\tau_0}' X_{2\tau_0}R -R' (X_{2\tau}'X_{2\tau})(XX')^{-1} (X_{2\tau_0}' X_{2\tau_0})R,
\end{align*}
which implies that 
\begin{align*}
\frac{1}{T}(Z_{2\tau}'M Z_{2\tau_0}) 
&=
(1-\tau_0) R'\Sigma_XR -(1-\tau)(1-\tau_0)R'\Sigma_XR+O_p(T^{-1/2})\\
&=
\tau(1-\tau_0)R'\Sigma_XR+O_p(T^{-1/2}).
\end{align*}
Therefore,
\begin{align*}
\frac{1}{T}
(Z_{2\tau}'MZ_{2\tau_0})'(Z_{2\tau}'M Z_{2\tau})^{-1}(Z_{2\tau}'MZ_{2\tau_0})
&=
\frac{\tau(1-\tau_0)^2}{1-\tau}
R'\Sigma_XR+O_p(T^{-1/2}).
\end{align*}
Finally, 
$ 
T^{-1}\big\{
(Z_{2\tau_0}'MZ_{2\tau_0})
-
(Z_{2\tau}'MZ_{2\tau_0})'(Z_{2\tau}'M Z_{2\tau})^{-1}(Z_{2\tau}'MZ_{2\tau_0})
\big\} 
$ equals to 
\begin{align*}
\bigg[\tau_0(1-\tau_0)- \frac{\tau(1-\tau_0)^2}{1-\tau}  \bigg] R'\Sigma_XR+O_p(T^{-1/2})
=
(\tau_0-\tau)\frac{1-\tau_0}{1-\tau}R'\Sigma_XR+O_p(T^{-1/2}).
\end{align*}

\end{proof}
\subsection{Proof of Proposition 3}

\PropAMax*
\begin{proof}[Proof of Proposition \ref{PropAMax}]
By definition, $Q(\tau_0)=\sigma_0^2$.
Note that $\delta_0' R' \Sigma_X R\delta_0  >0$. This is because (1) $R$ has full column rank, (2) $\delta_0 \ne 0$, and (3) $\Sigma_X$ is assumed to be positive definite. 
Hence, $Q(\tau)>\sigma_0^2$ $\forall \tau \ne \tau_0$. Recall that $A(\tau) = g(Q(\tau))$ where $g(x) = -\frac{1}{2} \log(x)$. Hence $A(\tau) = -\frac{1}{2} \log(\sigma_0^2)$ if $\tau=\tau_0$ and $A(\tau)<-\frac{1}{2} \log(\sigma_0^2)$ otherwise. 
\end{proof} 


\subsection{Proof of Proposition 4}
\PropACts*
\begin{proof}[Proof of Proposition \ref{PropACts}]
Recall that $A_T(\tau) = g(Q_T(\tau))$ and $A(\tau)=g(Q(\tau))$ where $g(x) = -\frac{1}{2}\log(x)$. By Taylor approximation, there is $c$ between $x$ and $a$:
\begin{align*}
g(x) -g(a) = g'(a)(x-a) + \frac{1}{2}g''(c) (x-a)^2.
\end{align*}
Hence, for each $\tau \in B^c_{M/T}(\tau_0)$, there is $c_T$ between $Q_T(\tau)$ and $Q(\tau)$:
\begin{align*}
g\left( Q_T(\tau) \right) - g\left( Q(\tau) \right)  
&=
g'\left( Q(\tau) \right) \big(   Q_T(\tau) - Q(\tau)   \big)
+
\frac{1}{2}g''\left( c_T \right) \big(   Q_T(\tau) - Q(\tau)   \big)^2\\
&=
g'\left( Q(\tau) \right)O_p(T^{-1/2})
+
O_p(T^{-1}),
\end{align*}
where we used Proposition \ref{PropLimitQ}. 
Similarly, there is $c_{0T}$ between $Q_T(\tau_0)$ and $Q(\tau_0)$:
\begin{align*}
g\left( Q_T(\tau_0) \right) - g\left( Q(\tau_0) \right)  
&=
g'\left( Q(\tau_0) \right)O_p(T^{-1/2})
+
O_p(T^{-1}).
\end{align*}
Note that $g''\left( c_{T} \right)=\frac{1}{2c_T^2} $ and $g''\left( c_{0T} \right) =\frac{1}{2c_{0T}^2}$ are bounded with probability tending to one because for each $\tau$, $Q_T(\tau)\overset{p}{\to} Q(\tau)$, and $Q(\tau) $ is bounded.  

Now, 
\begin{align*}
&\left\{  A_T(\tau) -A_T(\tau_0) \right\} -   \left\{  A(\tau) -A(\tau_0)    \right\}
=
\left\{  A_T(\tau) -A(\tau) \right\} -   \left\{  A_T(\tau_0) -A(\tau_0)    \right\} \\
&=
\bigg\{  g( Q_T(\tau) ) -g( Q(\tau)) \bigg\} -   \bigg\{  g( Q_T(\tau_0) ) -g(  Q(\tau_0) )    \bigg\} \\
&=
\bigg( g'\left( Q(\tau) \right)-g'\left( Q(\tau_0) \right)\bigg)O_p(T^{-1/2})
+
O_p(T^{-1})\\
 &=
\bigg[  
-\frac{1}{2(\sigma^2_0+\Delta(\tau) )}
-
\bigg( -\frac{1}{2( \sigma^2_0+\Delta(\tau_0) ) } \bigg)
 \bigg] 
O_p(T^{-1/2})
+
O_p(T^{-1})\\
&=
\frac{1}{2}\bigg[  \frac{1}{\sigma_0^2} - \frac{1}{\sigma_0^2+\Delta(\tau)}\bigg]
O_p(T^{-1/2})
+
O_p(T^{-1}).
\end{align*}

In general, there is $B>0$ such that  $1/b -1/(b+x) \leq Bx$ for $b,x>0$. Hence, $  1/\sigma_0^2 - 1/( \sigma_0^2+\Delta(\tau)) \leq B\Delta(\tau) \leq B'|\tau-\tau_0|$ where the last inequality holds for some $B'>0$ due to the shape of $Q(\tau)$.

Finally, 
\begin{align*}
 \frac{ | \left\{  A_T(\tau) -A_T(\tau_0) \right\} -   \left\{  A(\tau) -A(\tau_0)    \right\}  | }{ | \tau-\tau_0| }
\leq &
B'
O_p(T^{-1/2})
+
\frac{1}{| \tau-\tau_0|}
O_p(T^{-1})
\leq
O_p(T^{-1/2})
+
\frac{1}{M}O_p(1),
\end{align*}
for $|\tau-\tau_0| > M/T$.
The desired result is established by taking $M$ large enough.
\end{proof}  


\section{Derivation of posterior distributions under the normal likelihood and the conjugate prior }\label{derivation} 
In this section, we derive  posterior distributions under the normal likelihood and the conjugate prior. We have 
\begin{align}
p(\theta, \tau | \bm{D}_T ) 
&\propto 
p( \bm{Y}_T |  \bm{X}_T, \theta, \tau ) \pi(\theta, \tau)\nonumber \\
&\propto 
 \left(  \frac{1}{\sigma^2} \right)^{T/2} 
\exp\left[ -\frac{1}{2\sigma^2} \left\{ \sum_{t=1}^T \left( y_t - \chi_{\tau,t}' \gamma \right)^2   \right\} \right]\pi(\tau) \nonumber \\
&\times 
\left(  \frac{1}{\sigma^2} \right)^{\underline{a} + p/2 - 1} \exp\left[ -\frac{1}{\sigma^2} \left\{    \underline{b} + \frac{1}{2} \left( \gamma - \underline{\mu} \right)' \underline{H}   \left( \gamma - \underline{\mu} \right)  \right\}  \right] \nonumber \\
&\propto
\left(  \frac{1}{\sigma^2} \right)^{\underline{a} + (p+T)/2 - 1}
 \exp\left[ -\frac{1}{\sigma^2} \left\{    \bar{b}_\tau + \frac{1}{2} \left( \gamma - \bar{\mu}_\tau \right)' \bar{H}_\tau   \left( \gamma - \bar{\mu}_\tau \right)  \right\}  \right]\pi(\tau), \label{post_conj}
 \end{align}
where $p=d_x+d_z$. 
From above, we can deduce that 
\begin{align*}
\gamma \vert \sigma^2 \tau, \bm{D}_T &\sim N_p\left(\bar{\mu}_{\tau}, \sigma^2 \bar{H}_\tau^{-1} \right) \text{, and }\\
\sigma^2 \vert \tau, \bm{D}_T &\sim Inv Gamma \left( \bar{a},\bar{b}_\tau \right).
\end{align*}

Integrating the right hand side of \eqref{post_conj} with respect to $\gamma$, we obtain 
\begin{align}\label{sigma}
\left(  \frac{1}{\sigma^2} \right)^{\underline{a} + T/2 - 1}
\exp\left[ -\frac{1}{\sigma^2}   \bar{b}_\tau  \right]
\left[ \det \left( \bar{H}_\tau \right)  \right] ^{-0.5}
\pi(\tau). 
\end{align}
Integrating the above with respect to $\sigma^2$ over the positive part of the real line and using the change of variable $\phi = 1/\sigma^2$, we get the marginal posterior for $\tau$
\begin{align*}
\pi(\tau | \bm{D}_T ) \propto 
\left[ 
\det
\left( \bar{H}_\tau \right)
 \right]^{-0.5} 
\bar{b}_\tau^{-\bar{a}}  \pi(\tau),
\end{align*}
Finally, 
we apply the well-known property that the integral of a normal-inverse-gamma distribution with respect to $\sigma^2$ is a t-distribution to \eqref{post_conj} to conclude that 
\begin{align*}
\gamma
\big| \tau, \bm{D}_T \sim t_{p} 
\left(
2\bar{a}, 
\bar{\mu}_{\tau},
( \bar{b}_\tau /  \bar{a} )\bar{H}_\tau^{-1}
\right).
\end{align*}

\clearpage 

\bibliography{references}

\end{document}